\documentclass[runningheads]{llncs}
\pdfoutput=1
\usepackage[title]{appendix}
\usepackage{graphicx}
\usepackage{subcaption}
\usepackage{amsmath,amssymb,textcomp, amsthm}
\usepackage{enumerate}
\usepackage{algpseudocode}
\usepackage{algorithm}
\usepackage{wrapfig}
\usepackage{gensymb}
\usepackage{bold-extra}
\usepackage{mdframed}
\usepackage{thmtools,thm-restate}
\usepackage[htt]{hyphenat}

\begin{document}

\title{Voronoi diagram of orthogonal polyhedra\\ in two and three dimensions}
\titlerunning{Voronoi diagram of orthogonal polyhedra}
\author{Ioannis Z.~Emiris\inst{1,2}
\and Christina Katsamaki\inst{1}}
\authorrunning{I.Z.~Emiris and C.~Katsamaki}
\institute{Department of Informatics and Telecommunications\\
National and Kapodistrian University of Athens, Greece \and
ATHENA Research and Innovation Center, Maroussi, Greece\\
\email{\{emiris,ckatsamaki\}@di.uoa.gr}
}
\maketitle  
\begin{abstract}
Voronoi diagrams are a fundamental geometric data structure for obtaining proximity relations. 
We consider collections of axis-aligned orthogonal polyhedra in two and three-dimensional space under the max-norm, which is a particularly useful scenario in certain application domains. We construct the exact Voronoi diagram inside an orthogonal polyhedron with holes defined by such polyhedra. 
Our approach avoids creating full-dimensional elements on the Voronoi diagram and yields a skeletal representation of the input object. 
We introduce a complete algorithm in 2D and 3D that follows the subdivision paradigm relying on a bounding-volume hierarchy; this is an original approach to the problem. The complexity is adaptive and comparable to that of previous methods. Under a mild assumption it is
$O(n/ \Delta + 1/\Delta^2)$ in 2D or $O(n\alpha^2/\Delta^2 +1/\Delta^3)$ in 3D, where $n$ is the number of sites, namely edges or facets resp., $\Delta$ is the maximum cell size for the subdivision to stop, and $\alpha$ bounds vertex cardinality per facet.
We also provide a numerically stable, open-source implementation in Julia, illustrating the practical nature of our algorithm.

\end{abstract}

\section{Introduction}
Orthogonal shapes are ubiquitous in numerous applications including raster graphics and VLSI design. We address Voronoi diagrams of 2- and 3-dimensional orthogonal shapes. We focus on the $L_\infty$ metric which is used in the relevant applications and has been studied much less than $L_2$.

A \textit{Voronoi diagram} partitions space into regions based on distances to a given set $\mathcal{S}$ of geometric objects in $\mathbb{R}^d$. Every $s\in \mathcal{S}$ is a \textit{Voronoi site} (or simply a \textit{site}) and its \textit{Voronoi region} under metric $\mu$, is 
 $
 V_\mu(s) = \{x\in \mathbb{R}^d \mid \mu(s,x) < \mu(x,s'),\text{ }s' \in \mathcal{S} \setminus s \}.
 $ 
 The \textit{Voronoi diagram} is the set $\mathcal{V}_{\mu}(\mathcal{S}) = \mathbb{R}^d \setminus \bigcup_{s\in \mathcal{S}} V_{\mu}(s)$, consisting of all points that attain their minimum distance to $\mathcal{S}$ by at least two Voronoi sites.  For general input, the Voronoi diagram is a collection of faces of dimension $0,1,\dots , d-1$. A face of dimension $k $ comprises  points equidistant to at least $d+1-k$ sites. Faces of dimension 0 and 1 are called \textit{Voronoi vertices} and \textit{Voronoi edges} respectively.
The union of Voronoi edges and vertices is the \textit{1-skeleton}.
Equivalently, a Voronoi diagram is defined as the \textit{minimization diagram}
of the distance functions to the sites. The diagram is a partitioning of space into regions, each region consisting of points where some function has lower value than any other function.

\begin{wrapfigure}{r}{0.4\textwidth}
\vspace{-0.6cm}
\centerline{\includegraphics[scale=0.3]{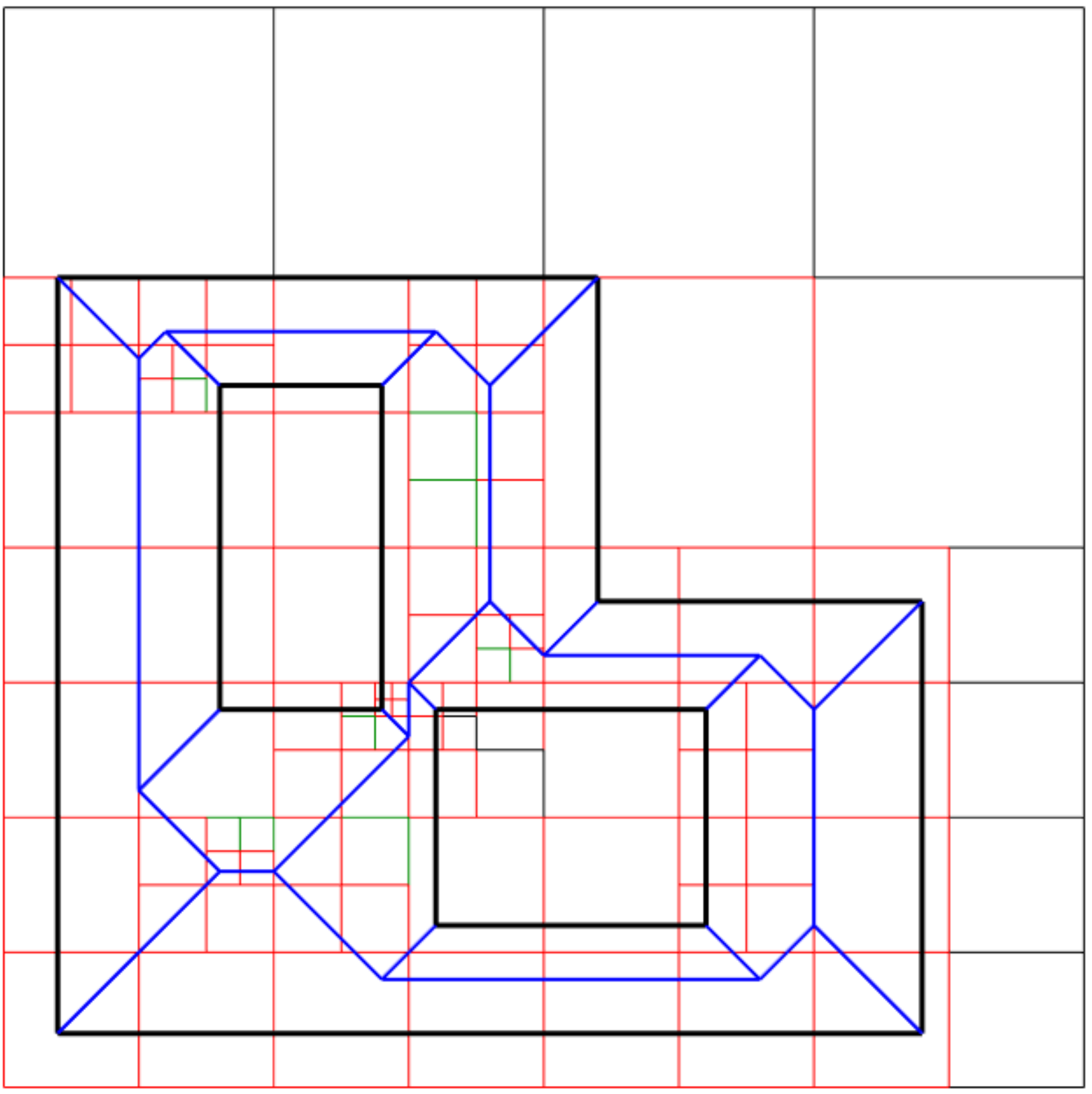}}
\caption{Voronoi diagram of a rectilinear polygon with 2 holes.\label{fig1}}
\vspace{-0.5cm}
\end{wrapfigure}

In this paper, we study Voronoi diagrams in the interior of an axis-aligned \textit{orthogonal polyhedron}; its faces meet at right angles, and the edges are aligned with the axes of a coordinate system. It may have arbitrarily high genus with holes defined by axis-aligned orthogonal polyhedra, not necessarily convex. Facets are simply connected (without holes) for simplicity.
The sites are the facets on the boundary of all polyhedra. 

The $L_\infty$ or Chebyshev distance is a metric defined on a vector space where the distance between two vectors is the greatest of their differences along any coordinate dimension. For two points $x,y \in \mathbb{R}^d$, with standard coordinates $x_{i}$ and $y_{i}$ respectively, their $L_\infty$ distance is 
$\mu_\infty(x,y) = \max_i \{ |x_i-y_i| \}$. The $L_\infty$ distance of $x$ to a set $S \subset \mathbb{R}^d$ is
$\mu_\infty(x,S) = \inf \{ \mu_\infty(x,y) \mid y\in S \}$.

The $\mathbf{L_\infty}$ \textbf{Voronoi diagram} uses the $L_\infty$ metric to measure distances.
In Figure~\ref{fig1}, the Voronoi diagram\footnote{computed by our software and visualized with {\tt Axl} viewer.} of a 
rectilinear polygon with 2 holes is shown in blue. Our algorithm follows the \textit{Subdivision Paradigm} and handles 2D and 3D sites. It reads in a region bounding all input sites 
and performs a recursive subdivision into cells (using quadtrees or octrees).
Then, a reconstruction technique is applied to produce an isomorphic representation of the Voronoi diagram. 

\subsubsection{Previous Work.}
If $V$ is the number of polyhedral vertices, the combinatorial complexity of our Voronoi diagrams  equals $O(V)$
in 2D \cite{PapLee01} and $O(V^2)$ in 3D \cite{straight3d}. In 3D, it is estimated experimentally to be, in general, $O(V)$ \cite{DBLP:journals/cvgip/MartinezGA13}.

Related work in 2D concerns $L_\infty$ Voronoi diagrams of segments. In \cite{PapLee01}, they introduce an $O(n\log n)$ sweep-line algorithm, where $n$ is the number of segments; they offer a robust implementation for segments with $O(1)$ number of orientations. 
Another algorithm implemented in library CGAL \cite{CDGP14} is incremental.
The $L_\infty$ Voronoi diagram of orthogonal polyhedra (with holes) is addressed in \cite{DBLP:journals/cvgip/MartinezGA13} in view of generalizing the sweep-line paradigm to 3D: in 2D it runs in $O(n\log n)$ as in \cite{PapLee01}, and in 3D the sweep-plane version runs in $O(kV)$, where $k=O(V^2)$ is the number of events. 

When the diagram is restricted in the interior of a polygon or polyhedron, it serves as a skeletal representation. A skeleton reduces the dimension of the input capturing its boundary's geometric and topological properties. In particular, \textit{straight skeletons} are very related to the $L_\infty$ Voronoi diagram of rectilinear polygons \cite{straight}. 
An algorithm for the straight skeleton of a simple polygon (not necessarily rectilinear) has complexity $O(V^{\frac{17}{11} +\varepsilon})$ for fixed $\varepsilon>0$ \cite{EppSkeleton}. For $x$-monotone rectilinear polygons, a linear time algorithm was recently introduced \cite{Held19}. In 3D, an analogous equivalence of the straight skeleton of orthogonal polyhedra and the $L_\infty$ Voronoi diagram exists \cite{straight3d} and a complete analysis of 3D straight skeletons is provided in \cite{Aurenhammer2016}. Specifically for 3D orthogonal polyhedra, in \cite{straight3d} they offer two algorithms that construct the skeleton in $O(\min\{V^2 \log V, k \log ^{O(1)} V\})$, where $k=O(V^2)$ is the number of skeleton features. Both algorithms are rather theoretical and follow a wavefront propagation process. Recently, the straight skeleton of a 3D polyhedral terrain was addressed \cite{Held18}.

For polyhedral objects it is common to consider as Voronoi sites the vertices, edges and higher-dimensional faces (when $d\ge 3$) that form their boundary \cite{polygon,Etzion:2002:CVS:511635.511636,Koltun:2002:PVD:513400.513428}. However, the combination of the underlying metric and the structure of the input polyhedron may cause the appearance of full-dimensional faces in the Voronoi diagram. Under this event, the Voronoi diagram cannot serve as a skeletal representation of the input shape and the design of an effective algorithm for its computation is significantly more difficult. For example, the $L_2$ Voronoi diagram of a polygon in the plane, under the standard definition, contains two dimensional regions of points equidistant to a reflex vertex and its adjacent edges. 
A usual approach towards eliminating the appearance of these two-dimensional regions is to consider as Voronoi sites the \textit{open} edges and the polygon vertices and define by convention 1-dimensional bisectors between them. A similar concept is this of defining the `cones of influence'' \cite{cone} or ``zones'' \cite{Yap-2012} of each site.

Full-dimensional faces on the Voronoi diagram in our setting are very frequent \cite{DBLP:journals/cvgip/MartinezGA13};
under $L_\infty$, when two points have same coordinate value, their bisector is full dimensional. Conventions have been adopted, to ensure bisectors between sites are not full-dimensional
\cite{PapLee01,DBLP:journals/cvgip/MartinezGA13,CDGP14}. We address this issue in the next section.

The literature on subdivision algorithms for Voronoi diagrams is vast, e.g.\ \cite{Bennett-2016,EMIRIS2013511,Yap-2012} and references therein. %\cite{Etzion:1999:CVD:304012.304029, Etzion:2002:CVS:511635.511636, EMIRIS2013511, Boada:2008:AGV:1450699.1450701, , Bennett-2016}. 
 Our work is closely related to \cite{Yap-2012,Bennett-2016}.
  These algorithms are quite efficient, since they adapt to the input,
and rather simple to implement. However, none exists for our problem.

\subsubsection{Our contribution.} 
We express the problem by means of the minimization diagram of a set of algebraic functions with restricted domain, that express the $L_\infty$ distance of points to the boundary.
The resulting Voronoi diagram, for general input, is $(d-1)-$dimensional.
We focus on 2D and 3D orthogonal polyhedra with holes,
where the resulting Voronoi diagram is equivalent to the straight skeleton. We introduce an efficient and complete algorithm for both dimensions, following the subdivision paradigm which is, to the best of our knowledge, the first subdivision algorithm for this problem. 
We compute the exact Voronoi diagram (since $L_\infty$ bisectors are linear). The output data structure can also be used for nearest-site searching. 

The overall complexity is output-sensitive, which is a major advantage. Under Hypothesis~1, which captures the expected geometry of the input as opposed to worst-case behaviour, the complexity is $O(n/\Delta+1/\Delta^2)$ in 2D, where $n$ the number of sites (edges) and $\Delta$ the separation bound (maximum edge length of cells that guarantees termination). This bound is to be juxtaposed to the worst-case bound of $O(n\log n)$ of previous methods.
In 3D, it is $O( n\, \alpha^2/ \Delta^{2} + 1/\Delta^3)$ where $\alpha$ bounds vertex cardinality per facet (typically constant). Under a further assumption (Remark~\ref{Rmk3d}) this bound becomes $O(V/\Delta^2 +1/\Delta^3)$ whereas existing worst-case bounds are quasi-quadratic or cubic in $V$. 
$\Delta$ is measured under appropriate scaling for the bounding box to have edge length~1. Scaling does not affect arithmetic complexity, but may be adapted to reduce the denominators' size in rational calculations.
The algorithm's relative simplicity has allowed us to develop a numerically stable software in Julia\footnote{\url{https://gitlab.inria.fr/ckatsama/L\_infinity\_Voronoi/}}, a user-friendly platform for efficient numeric computation; it consists of about 5000 lines of code and is the first open-source code in 3D.

The rest of this paper is organized as follows. The next section provides structural properties of Voronoi diagrams. In Sect.~\ref{Splanar} we introduce our 2D algorithm: the 2D and 3D versions share some basic ideas which are discussed in detail in this section. In particular, we describe a hierarchical data structure of bounding volumes, used to accelerate the 2D algorithm for certain inputs and is necessary for the efficiency of the 3D algorithm. Then we provide the complexity analysis of the 2D algorithm. In Sect.~\ref{S3d} we extend our algorithm and analysis to 3D. In Sect.~\ref{SImpl} we conclude with some remarks, examples and implementation details. %Due to space limitations, omitted proofs appear in the full version of our paper\footnote{\url{https://arxiv.org/abs/1905.08691}}.
%Some examples and implementation details are in the Appendix.

%%%%%%%%%%%%%%%%%%%%%%%%%%%%%%%%%%%%%%%%%%
\section{Basic definitions and properties}\label{Sbasics}

We introduce useful concepts in general dimension.
Let $\mathcal{P}$ be an orthogonal polyhedron of full dimension in $d$ dimensions, whose boundary consists of $n$ \textit{simply connected} (without holes) facets; these are edges or flats in 2D and 3D, resp. Note that $\mathcal{P}$ includes the shape's interior and boundary. Now $\mathcal{S}$ consists of the \textit{closed facets} that form the boundary of $\mathcal{P}$ including all facets of the interior polyhedra. There are as many such polyhedra as the genus.

Under the $L_\infty$, the Voronoi diagram $\mathcal{V}_\infty(\mathcal{S})$ often contains full-dimensional regions (e.g. Figure~\ref{fig:2dbisa}). Aiming at a $(d-1)$-dimensional Voronoi diagram, we will define an appropriate set of distance functions and restrict their minimization diagram to $\mathcal{P}$.

Let $V_\infty(s)$ denote the Voronoi region of site $s$ under the $L_\infty$ metric. Lemma~\ref{lem:prop} gives a property of standard $L_\infty$ Voronoi diagram preserved by Def.~\ref{DourVD}.

\begin{figure}[h]
\centering
\begin{subfigure}{0.4\textwidth}
\centering
\includegraphics[scale=0.6]{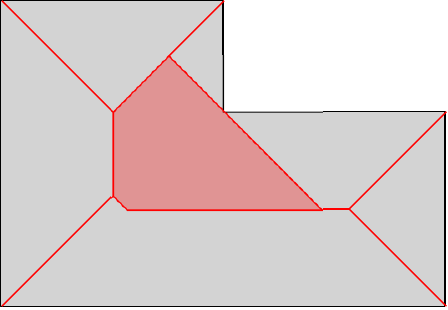}
\caption{ }
\label{fig:2dbisa}
\end{subfigure}
\begin{subfigure}{0.4\textwidth}
\centering
\includegraphics[scale=0.6]{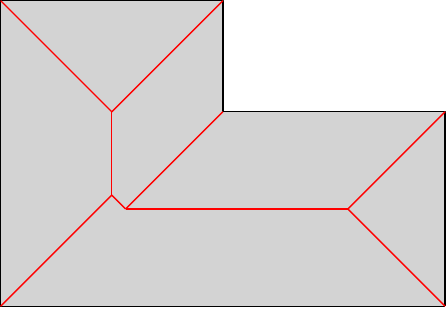}
\caption{ \label{fig:2dbisb}}
\end{subfigure}
\caption{Voronoi diagrams (in red): (a) standard, under $L_\infty$, (b) under Def.~\ref{DourVD}.\label{fig:2dbis}}
\end{figure}
%\vspace{-0.5cm}
\begin{restatable}{lemma}{prop} \label{lem:prop}
Let $s\in \mathcal{S}$. For every point $p \in V_\infty(s)$ it holds that $\mu_\infty(p,s) = \mu_\infty(p, \text{aff}(s))$, where aff$(s)$ is the affine hull of $s$.
\end{restatable}

\begin{proof}
Assume w.l.o.g. that $s \subset \{ x\in \mathbb{R}^d \mid x_k = c \}$, $k\in [d]$ and $c \in \mathbb{R}$. If $\mu_\infty(p,s) \neq \mu_\infty(p, \text{aff}(s))$, then $\mu_\infty(p, s) = \inf\{ \max_{i\in [d]\setminus k}\{|p_i-q_i|\} \mid \forall q\in s\}$ and there is $q\in \partial s$ such that $\mu_\infty(p,s)=\mu_\infty(p,q)= |p_j-q_j| $, $j\in [d]\setminus k $. To see this, suppose on the contrary that $q$ is in the interior of $s$. Then, we can find $q'\in s$ $\varepsilon$-close to $q$ such that $|p_j-q_j'|=|p_j-q_j|-\varepsilon \Rightarrow \mu_\infty(p,q') = \mu_\infty(p, q) - \varepsilon$, for any $\varepsilon>0$. This leads to a contradiction. 
Therefore,  there is a site $s'\neq s$ with $q\in s'$. %Then, $\mu_\infty(p, s') \le \mu_\infty(p, q)$.
Since $p \in V_\infty(s)$, then
$ \mu_\infty(p, s) <  \mu_\infty(p, s')  \le \mu_\infty(p, q)$; contradiction. 
\end{proof}

\begin{wrapfigure}{R}{0.3\textwidth}
\vspace{-0.5cm}\begin{center}
\includegraphics[scale=0.6]{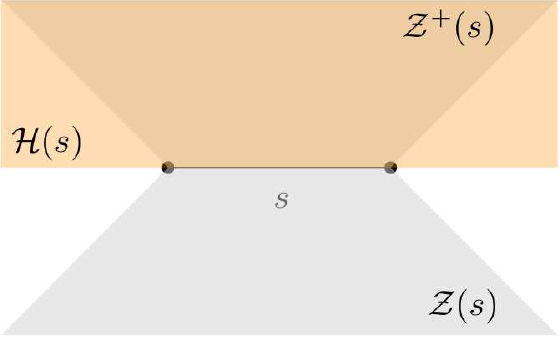}
\end{center}
\caption{$\mathcal{H}(s)$, $\mathcal{Z}(s)$, $\mathcal{Z}^+(s)$ for segment $s$.\label{fig3}}
\vspace{-0.5cm}
\end{wrapfigure}

 For $s\in \mathcal{S}$ let $\mathcal{H}(s)$ be the \textit{closed} halfspace of $\mathbb{R}^d$ induced by aff$(s)$ such that for every $p\in s$ there exists a point $q \in \mathcal{H}(s)$ : $q \in int(\mathcal{P})$ and $\mu_\infty(p, q)<\epsilon$, $\forall \epsilon>0$.
We define the \textbf{(unoriented) zone} of $s$ as $ \mathcal{Z}(s) :=\{ p \in \mathbb{R}^d  \mid \mu_\infty(p,s) = \mu_\infty(p,\text{aff}(s)) \}$.  The \textbf{oriented zone} of $s$ is $\mathcal{Z}^+(s) := \mathcal{H}(s) \cap \mathcal{Z}(s)$ (Figure~\ref{fig3}). % for an illustration in 2D).
We associate to $s$ the distance function 
$$
D_s(\cdot): \mathbb{R}^d \rightarrow \mathbb{R}: 
p\mapsto \begin{cases} \mu_\infty(p,s), \text{ if } p \in \mathcal{Z}^+(s), \\ \infty, \text{ otherwise.} \end{cases}
$$
The minimization diagram of $\mathcal{D}= \{D_s \mid s \in \mathcal{S} \}$ restricted to $\mathcal{P}$ yields a Voronoi partitioning. The \textit{Voronoi region} of $s$ with respect to $D_s(\cdot)$ is 
$$
V_{\mathcal{D}}(s) = \{ p \in \mathcal{P}  \mid D_s(p) < \infty \text{ and } \forall s' \in \mathcal{S}\setminus s \mkern9mu D_{s}(p)<D_{s'}(p) \}.
$$

\begin{definition}\label{DourVD}
The \textit{Voronoi diagram} of $\mathcal{P}$ w.r.t. $\mathcal{D}$ is $\mathcal{V}_\mathcal{D}(\mathcal{P})=\mathcal{P}\setminus \bigcup_{s\in \mathcal{S}} V_\mathcal{D}(s)$.
\end{definition}

This means one gets the Voronoi diagram of Figure~\ref{fig:2dbisb}. Clearly $\mathcal{P} \subset \bigcup_{s\in \mathcal{S}} \mathcal{Z}^+(s)$ (Figure~\ref{fig:2dbisa}). Denoting by $\overline{X}$ the closure of a set $X$, then $V_\infty(s) \subseteq V_\mathcal{D}(s) \subseteq \overline{V_\infty(s)} \subseteq \mathcal{Z}^+(s)$. The bisector of $s,s' \in \mathcal{S}$ w.r.t.\ $\mathcal{D}$ is $\text{bis}_{\mathcal{D}}(s,s') = \{ x\in \mathbb{R}^d \mid D_s(x) = D_{s'}(x) <\infty\}$. Then $\text{bis}_{\mathcal{D}}(s,s') \subset \text{affbis}(s,s')$, where $\text{affbis}(s,s')$ denotes the $L_\infty$ (affine) bisector of $\text{aff}(s), \text{aff}(s')$. In 2D (resp.\ 3D) if sites have not the same affine hull, bisectors under $\mathcal{D}$ lie on lines (resp.\ planes) parallel to one coordinate axis (resp.\ plane) or to the bisector of two perpendicular coordinate axes (resp.\ planes). 
%As long as sites have not the same affine hull, bisectors under $\mathcal{D}$ lie on hyperplanes parallel to one of the coordinate hyperplanes or to the bisector of two perpendicular coordinate hyperplanes. 
Although the latter consists of two lines (resp.\ planes), $\text{bis}_{\mathcal{D}}$ lies only on one, and it can be uniquely determined by the orientation of the zones. 

Degeneracy of full-dimensional bisectors, between sites with the same affine hull, is avoided by infinitesimal perturbation of corresponding sites. This is equivalent to assigning priorities to the sites; the full dimensional region of the former diagram is `to the limit' assigned to the site with the highest priority (Figure~\ref{degen3ab}). Such a perturbation always exists, both for 2D \cite[Lem.~13]{DBLP:journals/cvgip/MartinezGA13} and 3D \cite[Lem.~31]{DBLP:journals/cvgip/MartinezGA13}.

\begin{figure}
\centering
\begin{subfigure}{.45\textwidth}
  \centering
  \includegraphics[scale=0.5]{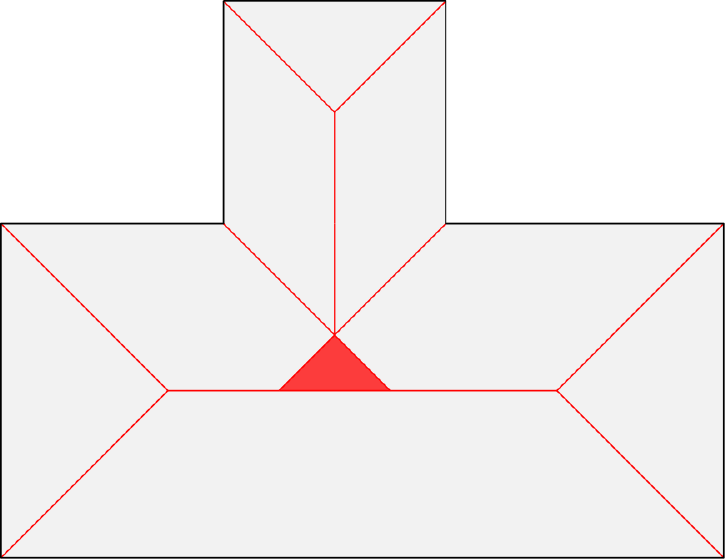}
  \caption{ }
  \label{degen3aa}
\end{subfigure}%
\hspace{0.2cm}
\begin{subfigure}{.45\textwidth}
  \centering
  \includegraphics[scale=0.5]{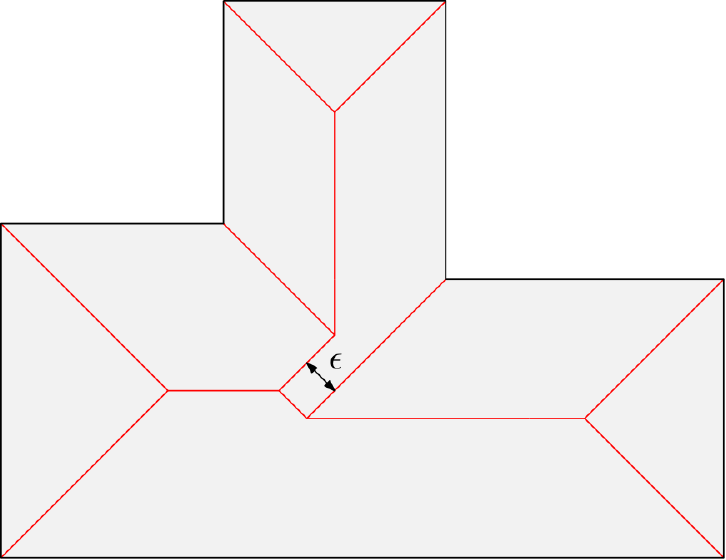}
 \caption{ \label{degen3ab}}
\end{subfigure}
\caption{(a) 2D Voronoi diagram for polygon with colinear edges. (b) 1D Voronoi diagram after infinitesimal perturbation of edges, where $\epsilon\rightarrow 0^+$.}
\label{degen3a}
\end{figure}

Set $X$ is \textit{weakly star shaped} with respect to $Y \subseteq X$ if $\forall x \in X$, $\exists y\in Y$ such that the segment $(x,y)$ belongs to $X$.

\begin{restatable}{lemma}{star}\label{star}
For every $s \in S$, $\overline{V_\mathcal{D}(s)}$ is weakly star shaped with respect to $s$.
\end{restatable}

\begin{proof}
Let $p \in \overline{V_\mathcal{D}(s)}$ and $\rho= D_s(p)= \mu_\infty(p,q)$, $q\in s$. The open ball $B_\infty(p,\rho)$ centered at $p$ with radius $\rho$ is empty of sites. For $t\in (0,1)$, let $w= tp+(1-t)q$ on the line segment $(p,q)$. 
Then, since for every $i\in [d]$ it is $|w_i-q_i| = t|p_i-q_i|$, it holds that $w\in \mathcal{Z}^+(s)$ and $D_s(w)= t\rho$. If $w \not \in \overline{V_\mathcal{D}(s)}$ there is a site $s'$ such that $D_{s'}(w) < t\rho$. But $B_\infty(w,t\rho) \subseteq B_\infty(p,\rho)$ and $s'$ intersects $B_\infty(p,\rho)$, leading to a contradiction. 
\end{proof}

%\begin{restatable}{corollary}{simply}
%\label{simply}
%$\overline{V_D(s)}$ is a \textit{simply connected} set
%\end{restatable}

Therefore, since every $s$ is simply connected, from Lemma~\ref{star} $\overline{V_{\mathcal{D}}(s)}$ is \textit{simply connected} and $V_{\mathcal{D}}(s)$ is also simply connected. 

Let the \textit{degree} of a Voronoi vertex be the number of sites to which it is equidistant. If the degree is $>d+1$, the vertex is \textit{degenerate}. 
Lemma~\ref{Ldegree} is nontrivial: in metrics like $L_2$ % detecting a degenerate vertex is hard, for 
degree is arbitrarily large.  For $d=2,3$ this bound is tight \cite{DBLP:journals/cvgip/MartinezGA13}.

\begin{restatable}{lemma}{degree}\label{Ldegree}
(a) The maximum degree of a Voronoi vertex is less than or equal to $2^d d$. (b) When $d=2$, a Voronoi vertex cannot have degree 7.
\end{restatable}

\begin{proof}
(a) Consider the vertex placed at the origin; $2^d$ orthants are formed around the vertex. To obtain the maximum number of Voronoi regions in each orthant, we count the maximum number of %Voronoi regions that can be contained in each orthant and share this Voronoi vertex. Equivalently, we count the number of 
Voronoi edges in the interior of an orthant that have this Voronoi vertex as endpoint; at most one such edge can exist in each orthant. Since these Voronoi edges are equidistant to $d$ sites, result follows.
%and the edge is in the interior of the orthant, there is a site among the $d$ that is parallel to each coordinate hyperplane.  

(b) %When $d=2$, the Voronoi edge in the intersection of two closed Voronoi regions, when parallel to the $x$-axis (resp. $y$-axis), is on the bisector of sites both parallel to the $x$-axis (resp. $y$-axis). When the Voronoi edge is rotated by 45\textdegree, the Voronoi sites are perpendicular among each other.
Let $v^*= (x^*, y^*)$ be a  Voronoi vertex of degree 7. Since 7 Voronoi edges meet at $v^*$, due to symmetry, we examine the two cases of Figure~\ref{degen1c}. When the configuration of Voronoi regions around the vertex is like in Figure~\ref{degen1c}a, % from the above remark we deduce that 
then $s_1$ is a horizontal segment and $s_2, s_7$ are vertical. Then,  $\text{aff}(s_2), \text{aff}(s_7) \subset \{(x,y) \in \mathbb{R}^2 \mid x > x^*\}$. Since $v^*\in \mathcal{Z}^+(s_2) \cap \mathcal{Z}^+(s_7)$ and is equidistant to both $s_2$ and $s_7$, the affine hulls of $s_2, s_7$ coincide. Then, whichever is the orientation of $s_1$, the affine bisectors of $s_1, s_2$ and $s_1, s_7$ cannot meet like in Figure~\ref{degen1c}a.
 When like in Figure~\ref{degen1c}b, since $b_3$ is vertical, $s_1$ is vertical. But since $b_1$ is horizontal, $s_1$ must be horizontal; a contradiction. 
\end{proof}

\begin{figure}[h]
\centering
\begin{subfigure}{.45\textwidth}
  \centering
  \includegraphics[scale=0.6]{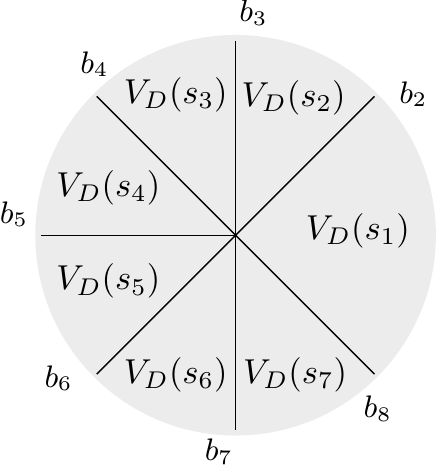}
  \caption{ \label{fig:deg7}}
\end{subfigure}%
\hspace{0.2cm}
\begin{subfigure}{.45\textwidth}
  \centering
  \includegraphics[scale=0.6]{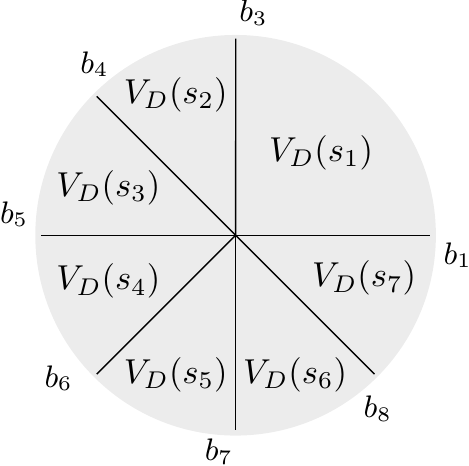}
 \caption{ \label{fig:deg7b}}
\end{subfigure}
\caption{The two cases in proof of Lemma \ref{Ldegree}. }%(b) where (a) $b_1$ is omitted, (b) $b_2$ is omitted.}
\label{degen1c}
\end{figure}

\section{Subdivision algorithm in two dimensions}\label{Splanar}

Given manifold rectilinear polygon $\mathcal{P}$, i.e.\ every vertex being shared by exactly two edges, the input consists of $\mathcal{S}$ and a box $\mathcal{C}_0$ bounding $\mathcal{P}$. Non-manifold vertices can be trivially converted to manifold with an infinitesimal perturbation. Subdivision algorithms include two phases. 
First, recursively subdivide $\mathcal{C}_0$ to 4 identical cells until certain criteria are satisfied, and the diagram's topology can be determined in $O(1)$ inside each cell. 
The diagram is reconstructed in the second phase.

\subsection{Subdivision Phase}
We consider subdivision cells as closed. Given cell $\mathcal{C}$, let $\phi(\mathcal{C})$ be the set of sites whose closed Voronoi region intersects $\mathcal{C}$: $ \phi(\mathcal{C})= \big\{ s \in \mathcal{S} \mid \overline{V_{\mathcal{D}}(s)} \cap \mathcal{C} \neq \emptyset \big\}$. 
For point $p \in \mathcal{P}$ we define its \textbf{label set}
$\lambda(p) = \{s \in \mathcal{S} \mid p \in  \overline{V_{\mathcal{D}}(s)} \}$.
When $p \in \mathcal{P} ^c$, where $\mathcal{P}^c$ is the complement of $\mathcal{P}$, then $\lambda(p)= \emptyset.$ 
Using the definition of label sets, one can alternatively write $\phi(\mathcal{C})$ as
$\phi(\mathcal{C}) = \bigcup_{p \in \mathcal{C}} \lambda(p) $.

Intuitively, storing $\phi(\mathcal{C})$ for every cell of the subdivision would ensure that upon termination and when the cardinality of this set for every leaf cell is a small constant, we could determine the topology of the Voronoi diagram inside each of them in constant time. %However, the computational cost of this operation is prohibitive. 
The computation of $\phi(\mathcal{C})$ is hereditary, since $\phi(\mathcal{C}) \subseteq \phi(\mathcal{C}')$, if $\mathcal{C}'$ is the parent of $\mathcal{C}$. But it is rather costly; given $\phi(\mathcal{C}')$ with $|\phi(\mathcal{C}')| =\kappa$, it takes $O(\kappa^2)$ to compute $\phi(\mathcal{C})$, since the relative position of $\mathcal{C}$ to the bisector of every pair of sites in $\phi(\mathcal{C}')$ must be specified. 
Alternatively,  following the work of \cite{Yap-2012,Bennett-2016}, instead of implementing the exact predicate $\phi(\cdot)$, we compute an approximate one. We denote by  $p_{_\mathcal{C}}$, $r_{_\mathcal{C}}$ the center and the $L_\infty$-radius of $\mathcal{C}$ and define the \textbf{active set} of $\mathcal{C}$ as:
$$
\widetilde{\phi}(\mathcal{C}):=\big\{ s\in \mathcal{S}  \mid  \mathcal{Z}^+(s) \cap \mathcal{C} \neq \emptyset, \text{ and } \mu_\infty(p_{_\mathcal{C}}, s) \le2 r_{_\mathcal{C}} + \delta_{_\mathcal{C}}  \big\},
$$
\noindent where $\delta_{_\mathcal{C}} = \min_s D_s(p_{_\mathcal{C}})\text{, if } p_{_\mathcal{C}} \in \mathcal{P}$, and 0 otherwise.
%and $rad(\mathcal{C})$ stands for the radius of $\mathcal{C}$ when interpreted as the $L_\infty$ closed ball $\overline{B_\infty(p_{_\mathcal{C}}, rad(\mathcal{C}))}$.
We now explain how $\widetilde{\phi}$ approximates $\phi$ by adapting \cite[Lem.2]{Bennett-2016}, where $\widetilde{\phi}$ appears as a \textit{soft version} of $\phi$. 

\begin{restatable}{lemma}{yap}\label{lem:yap}
(a) For every cell $\mathcal{C}$, $\phi(\mathcal{C}) \subseteq \widetilde{\phi}(\mathcal{C})$.
(b) For a sequence of cells $(\mathcal{C})_i$ monotonically convergent to point $p\in \mathcal{P}$,  $\widetilde{\phi}(\mathcal{C}_i) = \phi(p)$ for $i$ large enough. 
\end{restatable}

\begin{proof}
%\begin{environment}\leavevmode
(a)  If $\phi(\mathcal{C})= \emptyset$, assertion follows trivially. Let $s\in \phi(\mathcal{C})$ and $p\in \mathcal{C} \cap \overline{V_D(s)}$. It holds that $D_s(p) \le D_{s'}(p) \Rightarrow \mu_\infty(p,s) \le \mu_\infty(p,s')$ for every $s'\in \mathcal{S}$. We distinguish two cases according to the position of $p_{_\mathcal{C}}$ relatively to $\mathcal{P}$. If $p_{_\mathcal{C}}\in \mathcal{P}$ and $p_{_\mathcal{C}} \in \overline{V_{\mathcal{D}}(s^*)}$, then:
\begin{align} \mu_\infty(p_{_\mathcal{C}}, s) & \le \mu_\infty(p_{_\mathcal{C}}, p) +\mu_\infty(p,s)  
  \le \mu_\infty(p_{_\mathcal{C}}, p) +\mu_\infty(p,s^*) \le\nonumber \\
 & \le 2\mu_\infty(p_{_\mathcal{C}}, p) +\mu_\infty(p_{_\mathcal{C}},s^*)  
  \le 2r_{_\mathcal{C}} +\mu_\infty(p_{_\mathcal{C}}, s^*). \nonumber
\end{align} 
\noindent Otherwise, if $p_{_\mathcal{C}}\not \in \mathcal{P}$, since $\mathcal{C}\cap \mathcal{P}\neq \emptyset$, there is a site $s'$ intersecting $\mathcal{C}$ such that $\mu_\infty(p,s') \le r_{_\mathcal{C}}$. Therefore,
$
\mu_\infty(p_{_\mathcal{C}}, s) 
\le \mu_\infty(p_{_\mathcal{C}}, p) +\mu_\infty(p,s) 
  \le \mu_\infty(p_{_\mathcal{C}}, p) +\mu_\infty(p,s') 
  \le  2r_{_\mathcal{C}}. 
$

(b) There exists $i_0\in \mathbb{N}$ such that for $i \ge i_0$ $\mathcal{C}_i\cap \mathcal{P} \neq \emptyset$. Therefore, for $i\gg i_0$, since $p_{_{\mathcal{C}_i}} \rightarrow p$ and $r_{_{\mathcal{C}_i}} \rightarrow 0$, for every $s \in \mathcal{C}_i$, (a) implies that $s\in \lambda(p)=\phi(p)$. Since $\phi(p) \subseteq \widetilde{\phi}(\mathcal{C}_i)$, result follows.

%\end{environment}
\end{proof}

The previous lemma is crucial in proving correctness of the algorithm. It is also worth mentioning that a preliminary version of $\widetilde{\phi}$ first appeared in \cite{Milenkovic93robustconstruction}.

One can easily verify $\widetilde{\phi}(\mathcal{C}) \subseteq \widetilde{\phi}(\mathcal{C}')$, therefore the complexity of computing $\widetilde{\phi}(\mathcal{C})$ is linear in the size of $\widetilde{\phi}(\mathcal{C}')$.  %Therefore, this predicate seems to offer a good trade-off between accurate and fast computation.
% A preliminary version of this predicate was introduced in \cite{Milenkovic93robustconstruction}. 
% and has been exploited numerous times \cite{Bennett-2016,Yap-2012,Sud2006HomotopyPA}.
The algorithm proceeds as follows: For each subdivision cell we maintain the label sets of its corner points and of its central point, and $\widetilde{\phi}$. The subdivision of a cell stops whenever at least one of the \textit{termination criteria} below holds (checked in turn). Upon subdivision, we propagate $\widetilde{\phi}$ and the label sets of the parent cell to its children. For every child we compute the remaining label sets and refine its active set. Let $M$ be the maximum degree of a Voronoi vertex $(M\le 8)$.

\subsubsection{Termination criteria:} 
% \begin{enumerate}\item[
(T1) $\mathcal{C} \subseteq V_{\mathcal{D}}(s)$ for some $s\in \mathcal{S}$;
(T2) $int(\mathcal{C} ) \cap \mathcal{P} = \emptyset$;
(T3) $|\widetilde{\phi}(\mathcal{C} ) | \le 3$;  
%\item[(T4)] $|\widetilde{\phi}(\mathcal{C} ) | = 4$ and the sites in $\widetilde{\phi}(\mathcal{C} ) $ share a vertex $v \in \mathcal{C}$
(T4) $|\widetilde{\phi}(\mathcal{C} ) | \le M$ and the sites in $\widetilde{\phi}(\mathcal{C} ) $ define a unique Voronoi vertex $v \in \mathcal{C}$.
%\item[(T6)] rad$(\mathcal{C} )\le \varepsilon$, where $\varepsilon$ is an input parameter (threshold).
% \end{enumerate}

When (T1) holds, $\mathcal{C}$ is contained in a Voronoi region so no part of the diagram is in it. (T2) stops the subdivision when the open cell is completely outside the polygon. If (T3) holds, we determine in $O(1)$ time the diagram's topology in $\mathcal{C}$ since there are $\le 3$ Voronoi regions intersected. (T4) stops cell subdivision if it contains a single degenerate Voronoi vertex.
The process is summarized in Alg.~\ref{algo1}.
% Note that (T4) could be omitted under the convention that the intersection point of 4 sites is a degenerate Voronoi vertex. Two Voronoi edges of zero length and two of positive length would meet at that vertex.

\renewcommand{\thealgorithm}{1}
\begin{algorithm}[h]
	\caption{Subdivision2D($\mathcal{P}$)}
	\label{algo1}
	\begin{algorithmic}[1]
		%\Require  {orthogonal polygon $\mathcal{P}$}
                    %  \Ensure {}
\State root $\gets$ bounding box of $\mathcal{P}$
\State $Q \gets$ root
\While {$Q \neq \emptyset$}
	\State $\mathcal{C} \gets$ pop$(Q)$
        \State Compute $\widetilde{\phi}(\mathcal{C})$ and the label sets of the vertices and the central point.
	\If {(T1) $\lor$ (T2) $\lor$ (T3) $\lor$ (T4)}
		\State \Return
	\Else 
\State Subdivide $\mathcal{C}$ into $\mathcal{C}_1, \mathcal{C}_2, \mathcal{C}_3, \mathcal{C}_4$
\State $Q \gets$ $Q \cup \{ \mathcal{C}_1, \mathcal{C}_2, \mathcal{C}_3, \mathcal{C}_4\}$
	\EndIf
\EndWhile
\end{algorithmic}
\end{algorithm}

\begin{restatable}{theorem}{thmsub}
\label{halt}
Algorithm \ref{algo1} halts.
\end{restatable}

\begin{proof}
Consider an infinite sequence of boxes $\mathcal{C}_{1}\supseteq \mathcal{C}_{2} \supseteq \dots$ such that none of the termination criteria holds. Since (T1) and (T2) do not hold for any $\mathcal{C}_i$ with $i \ge 1$, the sequence converges to a point $p \in \mathcal{V}_\mathcal{D}(\mathcal{P})$. From Lemma~\ref{lem:yap}(b), there exists $i_0\in \mathbb{N}$ such that $\widetilde{\phi}(\mathcal{C}_{i_0}) = \phi(p) = \lambda(p)$. Since $| \lambda(p)|\le 8$, (T4) will hold. 
\end{proof}

\begin{restatable}{lemma}{inpred}\label{term}
For a subdivision cell $\mathcal{C}$, let $v_1,\dots, v_4 $ its corner vertices. For $s\in \mathcal{S}$, $\mathcal{C}\subseteq V_{\mathcal{D}}(s)$ if and only if $v_1,\dots, v_4 \in V_{\mathcal{D}}(s)$.
\end{restatable}

\begin{proof}
Let $v_1,\dots, v_4 \in V_{\mathcal{D}}(s)$ and $p \in \mathcal{C}$. Then, $p \in \mathcal{Z}^+(s)$, since $\mathcal{Z}^+(s)$ is convex in 2D and $v_1,\dots, v_4 \in \mathcal{Z}^+(s)$. For $i=1,\dots, 4$ the open ball $B_i:=B_\infty(v_i, \mu_\infty(v_i, s))$ is empty of sites. Since $B_\infty(p, \mu_\infty(p,\text{aff}(s))) \subset \cup_{i\in[4]} B_i $ it holds that $\mu_\infty(p, \mathcal{P}) \ge \mu_\infty(p, \text{aff}(s)) = \mu_\infty(p, s)$.
On the other hand,  $\mu_\infty(p, \mathcal{P}) \le  \mu_\infty(p, s)$.
% and $\mathcal{C} \subseteq \overline{V_\infty(s)}$. 
So, if $p \not \in V_{\mathcal{D}}(s)$ there is a site $s'$ s.t. %$D_{s'}(p) \le D_{s}(p) \Rightarrow 
$D_{s'}(p) = D_{s}(p)$ and $p\in \mathcal{V}_\mathcal{D}(\mathcal{P})$. Therefore, since Voronoi regions are simply connected and Voronoi edges are straight lines, $p$ must be on the boundary of $\mathcal{C}$. 
The two possible configurations are shown in Figure~\ref{fig:in} and are contradictory; for the first, use an argument similar to that of Lemma~\ref{Ldegree}(b). For the second, notice that this cannot hold since the cell is square. %There is a choice of such $s'$ so that $\text{affbis}(s,s')$ intersects the interior of $\mathcal{C}$. Since $p$ lies on the affine bisector and is not a vertex of the cell, there exists $p'$ in the interior of $\mathcal{C}$ such that $r' \in \mathcal{Z}^+(s')$ and therefore $D_{s'}(p') = D_s(p')$. This leads to a contradiction. 
We conclude that $\mathcal{C}\subseteq V_{\mathcal{D}}(s)$.
The other direction is trivial. 
\end{proof}

\begin{figure}[h]
\centering
\begin{subfigure}{.45\textwidth}
  \centering
  \includegraphics[scale=0.53]{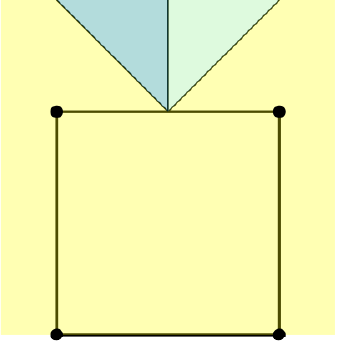}
  \caption{ \label{fig:in}}
\end{subfigure}%
\hspace{0.2cm}
\begin{subfigure}{.45\textwidth}
  \centering
  \includegraphics[scale=0.6]{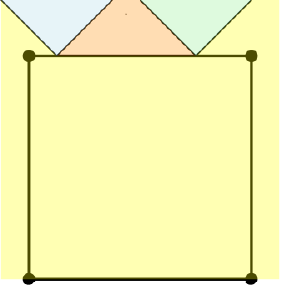}
 \caption{ \label{fig:inb}}
\end{subfigure}
\caption{The two cases in proof of Lemma \ref{term}. Different colors correspond to different Voronoi regions. }%(b) where (a) $b_1$ is omitted, (b) $b_2$ is omitted.}
\label{fig:in}
\end{figure}

\begin{lemma}
\label{lemma:out}
For a subdivision cell $\mathcal{C}$ it holds that $int(\mathcal{C})\cap \mathcal{P} =\emptyset$ if and only if $\lambda(p_{_\mathcal{C}}) = \emptyset$ and for every $s\in \widetilde{\phi}(C)$ it holds that $s\cap int(C) =\emptyset$.
\end{lemma}

\begin{proof}
Suppose that $\lambda(p_{_\mathcal{C}}) = \emptyset$, i.e. $p_{_\mathcal{C}} \not \in \mathcal{P}$, and that $s\cap int(C) =\emptyset$, for every $s\in \widetilde{\phi}(C)$. If is $int(\mathcal{C})\cap \mathcal{P}\neq \emptyset$, since the center of the cell is not in $\mathcal{P}$, a segment (site) on the boundary of $\mathcal{P}$ must intersect $int(\mathcal{C})$. This site belongs to $\widetilde{\phi}(C)$, leading to a contradiction. Proof of the other direction is trivial. 
\end{proof}

Hence one decides (T1) by checking the vertices' labels.
%Since $Z^+(s)$ is a convex set and $v_1,v_2,v_3,v_4 \in Z^*(s)$ we deduce that $\mathcal{C} \subseteq Z^*(s)$. 
(T2) is valid for $\mathcal{C}$ iff $\lambda(p_{_\mathcal{C}}) = \emptyset$ and $\forall s\in \widetilde{\phi}(\mathcal{C})$, $s\cap int(\mathcal{C}) =\emptyset$. Fot (T4), the presence of a Voronoi vertex in $\mathcal{C}$ is verified through constructor \texttt{VoronoiVertexTest}: given $\mathcal{C}$ with $|\widetilde{\phi}(\mathcal{C})| \ge 3$, the affine bisectors of sites in $\widetilde{\phi}(\mathcal{C})$ are intersected. If the intersection point is in $\mathcal{C}$ and in $\mathcal{Z}^+(s)$ for every $s\in \widetilde{\phi}(\mathcal{C})$ then it is a Voronoi vertex. We do not need to check whether it is in $\mathcal{P}$ or not; since (T1) fails for $\mathcal{C}$, if $v \not \in \mathcal{P}$, there must be $s$ intersecting $\mathcal{C}$ such that $v\not \in \mathcal{Z}^+(s)$: contradiction. 

\subsection{Reconstruction Phase}\label{3.3}

We take the quadtree of the subdivision phase and output a planar straight-line graph (PSLG) $G=(V,E)$ representing the Voronoi diagram of $\mathcal{P}$. $G$ is a (vertex) labeled graph and its nodes are of two types: \textit{bisector nodes} and \textit{Voronoi vertex nodes}. Bisector nodes span Voronoi edges and are labeled by the two sites to which they are equidistant. Voronoi vertex nodes correspond to Voronoi vertices and so are labeled by at least 3 sites. 

The reconstruction can be briefly described as follows:
we visit the leaves of the quadtree and, whenever the Voronoi diagram intersects the cell, bisector or vertex nodes are introduced. By connecting them accordingly with straight-line edges, we obtain the exact Voronoi diagram and not an approximation.  We process leaves with $|\widetilde{\phi}(\cdot)| \ge 2$ that do not satisfy (T1) nor (T2). 

\subsubsection{Cell with two active sites.}
 When $\widetilde{\phi}(\mathcal{C}) = \{s_1,s_2\}$, $\mathcal{C}$ intersects $V_{\mathcal{D}}(s_1)$ or $V_{\mathcal{D}}(s_2)$ or both; at this point $\mathcal{C}$ cannot be wholly contained in $V_D(s_1)$ or $V_D(s_2)$, since (T1) is not satisfied.  %$\mathcal{C}$ cannot be entirely contained in $V_D(s_1)$ or $V_D(s_2)$, since the cells that are processed do not satisfy (T1). 
The intersection of $\text{bis}_\mathcal{D}(s_1,s_2)$ with the cell, when non empty, is part of the Voronoi diagram: for each $p \in \text{bis}_\mathcal{D}(s_1,s_2) \cap \mathcal{C}$ it holds that $D_{s_1}(p) = D_{s_2}(p)$ and $\lambda(p) \subseteq \widetilde{\phi}(\mathcal{C}) = \{s_1,s_2\}$. Therefore $p \in \overline{V_\mathcal{D}(s_1)}\cap\overline{V_\mathcal{D}(s_2)}$. 

\begin{remark}\label{Ltwosites}
If there is no Voronoi vertex in $\mathcal{C}$ and $p_1, p_2 \in \text{bis}_\mathcal{D}(s_1,s_2) \cap \mathcal{C}$ for $s_1,s_2 \in \widetilde{\phi}(\mathcal{C})$, then $p_1p_2\subset \text{bis}_\mathcal{D}(s_1,s_2) $. 
\end{remark}

Since $\text{bis}_\mathcal{D}(s_1,s_2) \subset \text{affbis}(s_1,s_2)$ we intersect the affine bisector with the boundary of the cell. An intersection point $p \in bis_\infty(\text{aff}(s_1), \text{aff}(s_2))$ is  in $bis_{\mathcal{D}}(s_1,s_2)$ iff $p \in \mathcal{Z}^+(s_1) \cap \mathcal{Z}^+(s_2)$ (Lemma~\ref{Ltwosites}). 
If intersection points are both in $ \mathcal{Z}^+(s_1) \cap \mathcal{Z}^+(s_2)$ (Figure \ref{cornervertex2}), we introduce a bisector node in the middle of the line segment joining them, labeled by $\{s_1, s_2\}$. When only one intersection point is in $ \mathcal{Z}^+(s_1) \cap \mathcal{Z}^+(s_2)$, then $s_1, s_2$ must intersect in $\mathcal{C}$ (Figure \ref{cornervertex1}). We insert a bisector node at their intersection point labeled by $\{s_1, s_2\}$.

\begin{figure}
\begin{subfigure}{.5\textwidth}
  \centering
  \includegraphics{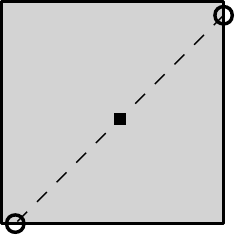}
  \caption{ }
  \label{cornervertex2}
\end{subfigure}%
\begin{subfigure}{.5\textwidth}
  \centering
  \includegraphics{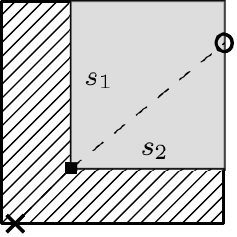}
 \caption{ }
  \label{cornervertex1}
\end{subfigure}
\caption{The intersection of the affine bisector and the cell. The circled points are in $ \mathcal{Z}^+(s_1) \cap \mathcal{Z}^+(s_2)$, whereas the crossed point is not. Square nodes are the bisector nodes inserted to the output graph. }
\label{cornervertex}
\end{figure}

\subsubsection{Cell with 3 active sites or more.}
When $|\widetilde{\phi}(\mathcal{C})|= 3$ and the \texttt{VoronoiVertexTest} finds a vertex in $\mathcal{C}$ or when $|\widetilde{\phi}(\mathcal{C})|\ge 4$ (a vertex has already been found), 
we introduce a Voronoi vertex node at the vertex, labeled by corresponding sites. In the presence of corners of $\mathcal{P}$ in $\mathcal{C}$, bisector nodes are introduced and connected to the vertex node (Figure~\ref{vertexnode}).

If no Voronoi vertex is in $\mathcal{C}$, we repeat the procedure described in previous paragraph for each pair of sites. Even if a bisector node is found, it is not inserted at the graph if it is closer to the third site. 

\vspace{0.2cm}
\begin{figure}
\centering
\begin{subfigure}{.3\textwidth}
  \centering
  \includegraphics{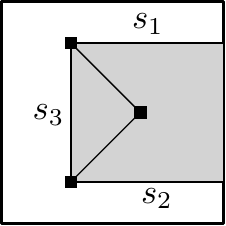}
  \caption{ }
  %\label{vertexnode}
\end{subfigure}%
\begin{subfigure}{.3\textwidth}
  \centering
  \includegraphics{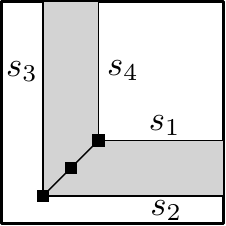}
 \caption{ }
  %\label{degenvertexnode}
\end{subfigure}
\caption{(a) A Voronoi vertex node labeled with $\{s_1, s_2, s_3\}$ connected with two bisector nodes labeled with $\{s_1,s_3\}$ and $\{s_2,s_3\}$. (b) A (degenerate) Voronoi vertex node labeled with $\{s_1, s_2, s_3, s_4\}$ connected with two bisector nodes. }
\label{vertexnode}
\end{figure}
\vspace{-0.5cm}
\subsubsection{Connecting the graph nodes.}
Once the graph nodes are fixed they have to be connected appropriately. The only graph edges added so far are those that are completely contained in a subdivision cell. The remaining graph edges must cross two subdivision cells.
We apply ``dual marching cubes" \cite{Scot-2005} to enumerate pairs of neighboring cells in time linear in the size of the quadtree: %(the \texttt{faceProc} and \texttt{edgeProc} procedures of the method are used). 
cells are neighboring if they share a facet. 

The traversal method involves two recursive functions: the \texttt{faceProc} and the \texttt{edgeProc} (see Figure \ref{dmc}). At first we call the \texttt{faceProc} for the bounding box of the input polygon. This function, for every internal node of the quadtree, recursively calls itself for each of its children and the \texttt{edgeProc} for every pair of neighboring children. When the \texttt{faceProc} reaches a leaf it terminates. When the \texttt{edgeProc} reaches two adjacent leaves, then the corresponding cells share a facet. Let $v_1, v_2$ be graph nodes in neighboring cells. We connect them if and only if:
\begin{itemize}
    \item  $v_1, v_2$ are bisector nodes and $\lambda(v_1)=\lambda(v_2).$
    \item $v_1$ is a bisector node, $v_2$ is a Voronoi vertex node and $\lambda(v_1)\subset \lambda(v_2)$.
    \item $v_1$, $v_2$ are Voronoi vertex nodes, $\lambda(v_1)\cap \lambda(v_2) = \{s, s'\} $ and $v_1v_2 \subset \mathcal{P}$.%the straight line segment $v_1v_2$ lies on the affine bisector of $s, s'$.
\end{itemize}

\noindent See Figure~\ref{fig:recex} for an example where $v_1,v_2$ are Voronoi vertex nodes with $\lambda(v_1)\cap \lambda(v_2) = \{s, s'\} $ and $v_1v_2 \not \subset \mathcal{P}$.

\begin{figure}
    \centering
    \includegraphics[scale=0.5]{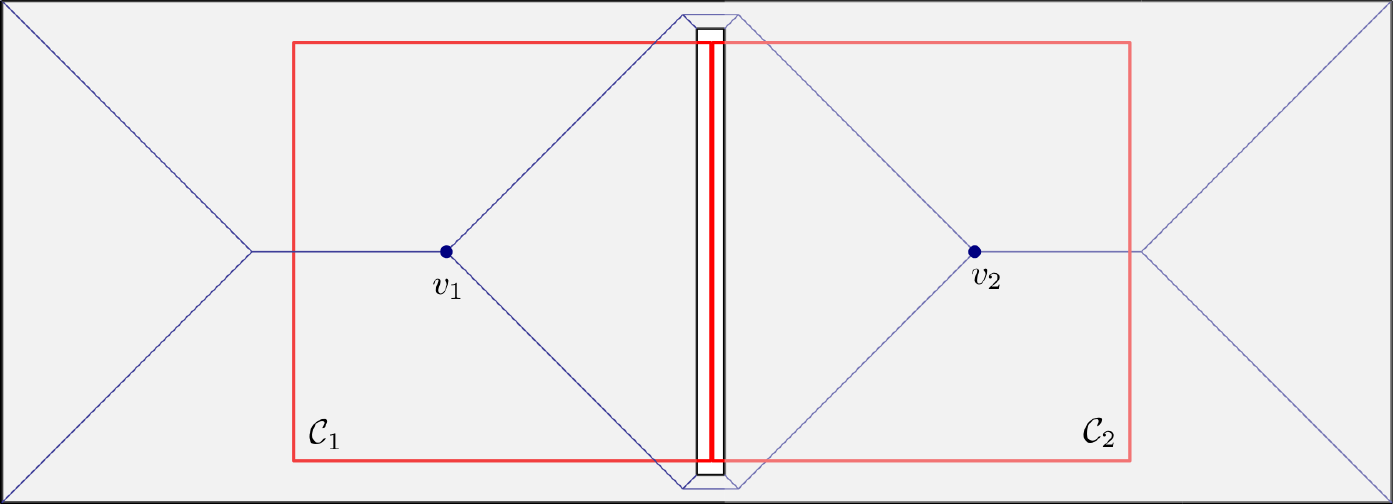}
    \caption{$\mathcal{C}_1, \mathcal{C}_2$ are two neighboring subdivision cells and the Voronoi vertices $v_1, v_2$ have two common sites as labels but are not connected with a Voronoi edge.}
    \label{fig:recex}
\end{figure}

%By performing a simple case analysis and using Lem.~\ref{Ltwosites} we conclude that:

\begin{restatable}[Correctness]{theorem}{reconstruction}
The output graph is isomorphic to $\mathcal{V}_\mathcal{D}(\mathcal{P}).$%the Voronoi diagram.
\end{restatable} 

\begin{proof}
We need to prove that the nodes in the graph are connected correctly. 
Let neighboring cells $\mathcal{C}_1, \mathcal{C}_2$ and $v_1, v_2$ graph nodes in each of them respectively. If $v_1, v_2$ are bisector nodes and $\lambda(v_1) = \lambda(v_2)$, then the line segment $v_1v_2$ is in $bis_{_\mathcal{D}}(s_1,s_2)$, for $s_1, s_2 \in \lambda(v_1) $, and on the Voronoi diagram (Rem.~\ref{Ltwosites}). 
If $v_1$ is a bisector node and $v_2$ is a Voronoi vertex node s.t. $\lambda(v_1) \subseteq \lambda(v_2)$, then  $v_1v_2 \subset \text{bis}_\infty(s_1,s_2)$. If the segment $v_1v_2$ is not on the Voronoi diagram then, there is a Voronoi vertex node different than $v_2$ in $\mathcal{C}_1$ or  $\mathcal{C}_2$; contradiction. 
At last, let $v_1$ and $v_2$ be Voronoi vertex nodes such that their labels have two sites in common, say $s, s'$, and the edge $v_1v_2\subset \mathcal{P}$. Vertices $v_1,v_2$ are both on the boundary of $\overline{V_\mathcal{D}(s)}\cap\overline{V_\mathcal{D}(s')}$. Since $v_1v_2\subset\mathcal{P}$, if it does not coincide with the Voronoi edge equidistant to $s, s'$, then both $v_1,v_2$ must also be on the boundary of a Voronoi region other than $\overline{V_\mathcal{D}(s)}$ and $\overline{V_\mathcal{D}(s')}$. This leads to a contradiction.  
\end{proof}

\begin{figure}[t]
\centering
\begin{subfigure}{\textwidth}

  \centering
  \includegraphics[scale=0.8]{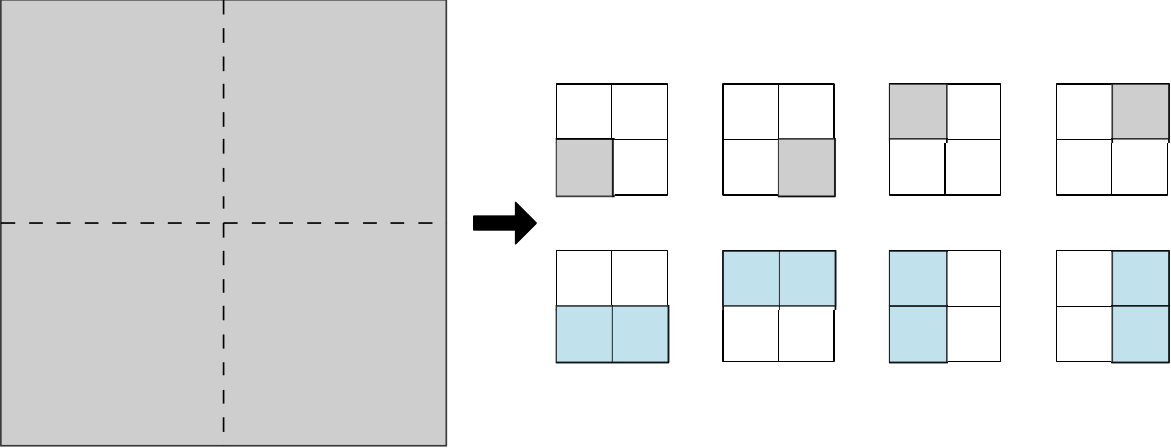}
  \caption{ }
  \label{dmc1}
\end{subfigure}%
 \vskip\baselineskip
\begin{subfigure}{\textwidth}
  \centering
  \includegraphics[scale=0.8]{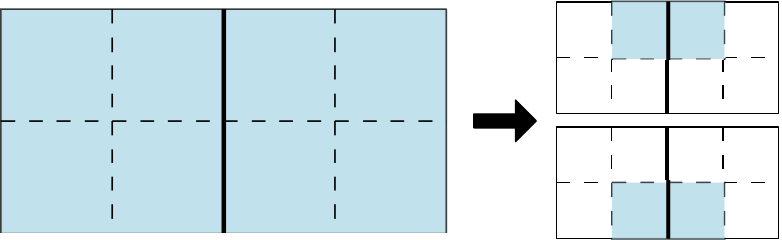}
 \caption{ }
  \label{dmc2}
\end{subfigure}
\caption{Illustration of (a) the \texttt{faceProc} function (b) the \texttt{edgeProc} function in the dual Marching Cubes method.}
\label{dmc}
\end{figure}

%\begin{proof}
%Let neighboring cells $\mathcal{C}_1, \mathcal{C}_2$ and $v_1, v_2$ graph nodes in each of them respectively. If $v_1, v_2$ are bisector nodes and $\lambda(v_1) = \lambda(v_2)$, then the line segment $v_1v_2$ is in $bis_{_\mathcal{D}}(s_1,s_2)$ and on the Voronoi diagram (Lem.~\ref{Ltwosites}). 
%If $v_1$ is a bisector node and $v_2$ is a Voronoi vertex node s.t. $\lambda(v_1) \subseteq \lambda(v_2)$, then  $v_1v_2 \subset bis_\infty(s_1,s_2)$. If not on the Voronoi diagram then, there is a Voronoi vertex node different than $v_2$ in $\mathcal{C}_1$ or  $\mathcal{C}_2$; contradiction. 
%At last, if $v_1$ and $v_2$ are Voronoi vertex nodes, then $v_1, v_2$ are adjacent iff their labels differ in two Voronoi sites, say $s_3, s_4$. They are connected as long as $v_1v_2$ coincides with the Voronoi edge equidistant to $s_3, s_4$. \qed
%\end{proof}

\subsection{Primitives, Data-structures, Complexity}\label{SSpredicates}

Assuming the input vertices are rational, Voronoi vertices are rational \cite{PapLee01}. Computing Voronoi vertices, and intersections between affine bisectors and cell facets require linear operations, distance evaluations and comparisons. Therefore, they are exact. %Moreover, our computations will involve intersecting affine bisectors with cell boundaries the coordinates of the graph nodes are computed exactly, by connecting them accordingly with straight-line edges, we obtain the exact Voronoi skeleton and not an approximation. 
 %Since the algebraic degree of computations is equal to one, it is a quite fair assumption. 
 The above operations, computing  $\widetilde{\phi}$ and deciding site-cell intersection are formulated to allow for a direct extension to 3D. In the sequel we discuss design of predicates, computation of label sets and construction of a Bounding Volume Hierarchy. We also provide a complexity analysis of the algorithm.

\subsubsection{Primitives. }

Membership in $\mathcal{H}(s)$ is trivial to decide, thus we focus on predicates that decide membership in $\mathcal{Z}(s)$.
Given $p\in \mathbb{R}^2$ and $s\in \mathcal{S}$, let $pr_{\text{aff}(s)}(p)$ the projection of $p$ to aff$(s)$ and $I_{p,s}$ the $1d-$interval on $\text{aff}(s)$ centered at $pr_{\text{aff}(s)}(p)$ with radius $\mu_\infty(p,\text{aff}(s))$.
\textbf{\texttt{inZone($p,s$)}} decides if 
$p\in \mathcal{Z}(s)$; this holds if and only if $I_{p,s} \cap s \neq \emptyset$ (Figure~\ref{inzone}). 

\begin{figure}[b]
\begin{center}
\includegraphics[scale=0.4]{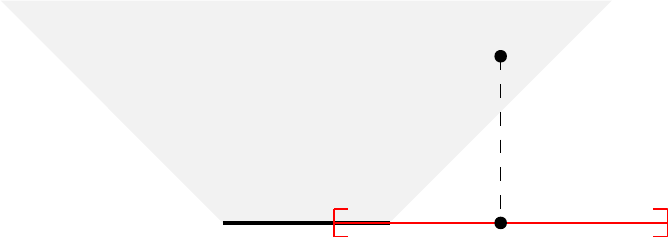}
\end{center}
\caption{Test performed by \texttt{inZone($p,s$)}.\label{inzone}}
\end{figure}

Given $\mathcal{C}$,$s\in \mathcal{S}$, \textbf{\texttt{ZoneInCell($s, \mathcal{C}$)}} decides if $\mathcal{Z}(s)\cap \mathcal{C} \neq \emptyset$. For this evaluation see Lemmma~\ref{zoneincell} and Figure~\ref{fig:zoneincell}.
%Given $p\in \mathbb{R}^2$ and $s\in \mathcal{S}$, let $pr_{\text{aff}(s)}(p)$ the projection of $p$ to aff$(s)$ and $I_{p,s}$ the $1d-$interval on $\text{aff}(s)$ centered at $pr_{\text{aff}(s)}(p)$ with radius $\mu_\infty(p,\text{aff}(s))$.
%\textbf{\texttt{inZone($p,s$)}} decides if 
%$p\in \mathcal{Z}(s)$; this holds iff $I_{p,s} \cap s \neq \emptyset$ (Fig.~\ref{inzone}). Predicate \textbf{\texttt{inOrientedZone($p,s$)}} decides if 
%$p \in \mathcal{Z}^+(s) $; if \texttt{inZone($p,s$)} is true we additionally check whether $p\in \mathcal{H}(s)$ or not. Membership in $\mathcal{H}(s)$ is trivial to decide.
%Given $\mathcal{C}$ \textbf{\texttt{OrientedZoneInCell($s, \mathcal{C}$)}} decides if $\mathcal{Z}^+(s)\cap \mathcal{C} \neq \emptyset$. For this evaluation see Lem.~\ref{zoneincell} and Fig.~\ref{fig:zoneincell}.

\begin{lemma}\label{zoneincell}
Let $s\in \mathcal{S}$, $f_1,f_2$ the two facets of  $\mathcal{C}$ parallel to aff$(s)$, $\rho_i = \mu_\infty(f_i, \text{aff}(s))$ for $i = 1,2$ and $p_{_\mathcal{C}}' = pr_{\text{aff}(s)}(p_{_\mathcal{C}}).$ %  We denote by $p_{_\mathcal{C}}'$ the projection of the center of $\mathcal{C}$ on the affine hull of $s$. 
Then, $\mathcal{Z}(s) \cap \mathcal{C} \neq \emptyset$ iff $\exists i\in \{1,2\}$ s.t.\ $B_\infty(p_{_\mathcal{C}}', r_{_\mathcal{C}}+\rho_i)\cap s \neq \emptyset$.
\end{lemma}

\begin{proof}
$\mathcal{Z}(s) \cap \mathcal{C} \neq \emptyset$ iff $\mathcal{Z}(s) \cap f_i \neq \emptyset$ for at least one $i\in \{1,2\}$: Let $p \in \mathcal{Z}(s) \cap C $ s.t. $p \not\in f_1\cup f_2$ and $\text{pr}_{f_i}(p)$ be the projection of $p$ on $f_i$. There exists $i\in\{1,2\}$ s.t.  $\mu_\infty (\text{pr}_{f_i}(p), \text{aff}(s)) > \mu_\infty(p, \text{aff}(s))$. % and  %$\mu_\infty(\text{pr}_{f_i}(p), s) = \mu_\infty(\text{pr}_{f_i}(p), \text{aff}(s))$ and 
 Then, $\text{pr}_{f_i}(p) \in \mathcal{Z}(s)$. %This always holds for at least one $i$. %Moreover, if only one, say $\text{pr}_{f_1}(p)$, is in $\mathcal{H}(s)$, the aff$(s)$ intersects the cell and $\mu_\infty(\text{pr}_{f_1}(p), \text{aff}(s)) > \mu_\infty(p, \text{aff}(s))$.
It holds that $\mathcal{Z}(s)\cap f_i \neq \emptyset$ iff $B_\infty(p_{_\mathcal{C}}', r_{_\mathcal{C}}+\rho_i)\cap s \neq \emptyset$: Let $q\in \mathcal{Z}(s)\cap f_i $ and $q'$ its projection on $\text{aff}(s)$. Then $\mu_\infty(q', s) \le \mu_\infty(q,\text{aff}(s))= \rho_i$ and $\mu_\infty(p_{_\mathcal{C}}',q') \le r_{_\mathcal{C}}$. We deduce that $B_\infty(p_{_\mathcal{C}}', r_{_\mathcal{C}}+\rho_i)\cap s \neq \emptyset$, since
$
\mu_\infty(p_{_\mathcal{C}}', s) \le \mu_\infty(p_{_\mathcal{C}}',q')+\mu_\infty(q',s) \le r_{_\mathcal{C}}+\rho_i.
$
For the inverse direction, %suppose that for some $i$
let $B_\infty(p_{_\mathcal{C}}', r_{_\mathcal{C}}+\rho_i) \cap s \neq \emptyset$ and $q'$ in $s$ s.t. $\mu_\infty(p_{_\mathcal{C}}',q') \le r_{_\mathcal{C}}+\rho_i$. Let $q$ be its projection on $\text{aff}(f_i)$. If $q\in f_i$ we are done. Otherwise, $q$ is at $L_\infty$ distance from $f_i$ equal to $\mu_\infty(p_{_\mathcal{C}}', q') - r_{_\mathcal{C}}$, attained at a boundary point $q''\in f_i$. Then,  
$
\rho_i \le \mu_\infty(q'',s) \le  \mu_\infty(q'',q') = \max \{ \rho_i, \mu_\infty(p_{_\mathcal{C}}', q') - r_{_\mathcal{C}} \}  = \rho_i.
$
It follows that $q'' \in \mathcal{Z}(s)$.%Since $f_i \subseteq \mathcal{H}(s)$, we conclude that $f_i\cap \mathcal{Z}^+(s)\neq \emptyset$. 
\end{proof}

\begin{figure}\centering
\begin{subfigure}{.45\textwidth}
  \centering
  \includegraphics[scale=0.5]{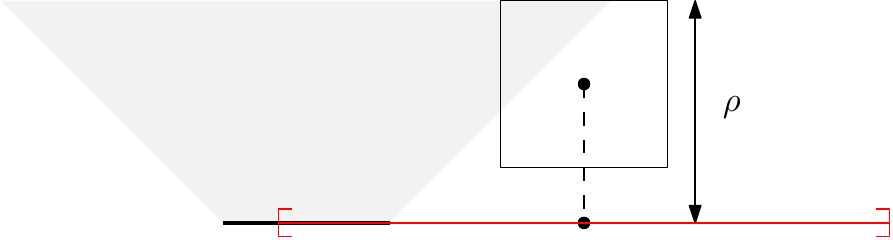}
  \caption{  \label{fig:zone1}}
\end{subfigure}%
\hspace{0.2cm}
\begin{subfigure}{.45\textwidth}
  \centering
  \includegraphics[scale=0.5]{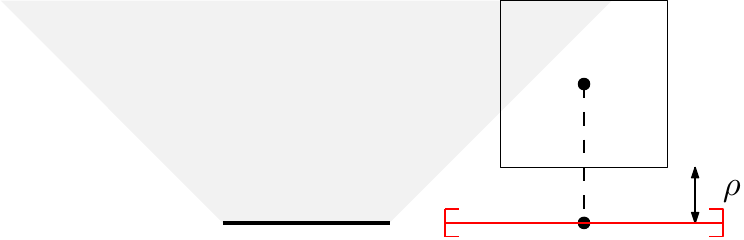}
 \caption{   \label{fig:zone2}}
\end{subfigure}
\caption{Illustration of test performed by \texttt{ZoneInCell} \label{fig:zoneincell}}
\end{figure}

To decide if $s\cap \mathcal{C}\neq \emptyset$ and if $s\cap int(\mathcal{C}) \neq \emptyset$, we use \textbf{\texttt{isIntersecting($s, \mathcal{C}$)}} and \textbf{\texttt{isStrictlyIntersecting($s, \mathcal{C}$)}} respectively. Design is trivial.  All these predicates are computed in $O(1)$ time.

Deciding whether a site belongs to $\widetilde{\phi}(\mathcal{C})$, is done using the predicates:
\begin{itemize}
\item \texttt{isIntersecting}: to decide whether the site intersects the cell centered at $p_{_\mathcal{C}}$ with radius $2\cdot r_{_\mathcal{C}} + \delta_{_\mathcal{C}} $. 
\item \texttt{ZoneInCell}: to decide whether the zone of the site intersects $\mathcal{C}$.
\end{itemize}

\paragraph{Computing label sets.}
% $\widetilde{\phi}$ interchangeably the current active set known for $\mathcal{C}$ (it may be the active set of its parent).
If $p\in\mathcal{P}\cap \mathcal{C}$ then its closest sites are in $\widetilde{\phi}(\mathcal{C})$. Deciding if $p\in \mathcal{P}$ is done by \texttt{LocationTest}, which identifies position based on the sites that intersect $\mathcal{C}$: among these we select those with minimum $L_\infty$ distance to $p$ and for whom \texttt{inZone($p,s$)} is true. If a convex (resp.\ concave) corner with respect to the interior of $\mathcal{P}$ is formed by these sites then $p\in \mathcal{P}$ iff it belongs to the intersection (resp.\ union) of the oriented zones. If no corner is formed or even if $\mathcal{C}$ is not intersected by any site, decision is trivial. This takes $O(|\widetilde{\phi}(\mathcal{C})|)$ time. In more details, we define:
\begin{itemize}
\item $T := \{ s\in \mathcal{S}\mid s\cap \mathcal{C} \neq \emptyset \}$
\item $d(p, T) := \min\{ \mu_\infty(p,s) \mid s\in T\text{ and } p\in \mathcal{Z}(s) \}$ 
\item $R(p, T):= \{s \in T \mid \mu_\infty(p,s) = d(p)  \text{ and } p \in \mathcal{Z}(s)\}$
\item $R'(p,s, T) = \{ s'\in R(p,T) \mid s\cap s' \neq \emptyset \}$, $s\in R(p, T)$.

\end{itemize}

If $T = \emptyset $ then the cell is contained in $\mathcal{P}$. Otherwise, even if $p \not \in \mathcal{P}$ there always exists a site $s$ intersecting $\mathcal{C}$ with $p \in \mathcal{Z}(s)$. Therefore $d(p,T)$ is well defined and $|R(p, T)| \ge 1$ for every $p \in \mathcal{C}$. We pick a site $s\in R(p,T)$ and then perform the \texttt{LocationTest} for $R'(p,s, T)$ that returns true if and only if $p\in \mathcal{P}$; this test identifies the position of $p$ relative to $\mathcal{P}$ using one site if $|R'(p,s, T)| = 1$ or a corner formed by two sites in $R'(p,s,T)$ otherwise. Note that $R'(p,s,T)$ may contain 1 or 2 sites.

\begin{algorithm}[!htp]
\texttt{LocationTest}$(p,R'(p,s,T))$
\begin{algorithmic}[1]
\If {$|R'(p,s,T)|=1$} \Comment $R' (p,s,T)= \{s\}$
	\State \Return  $p\in \mathcal{Z}^+(s_1)$
\ElsIf {$|R'(p,s,T)|=2$}  \Comment $R'(p,s)= \{s, s'\}$
	\If {$\angle s's$ is convex}
	\State \Return $p\in \mathcal{Z}^+(s)\cap \mathcal{Z}^+(s')$
	\Else
	\State \Return $p\in \mathcal{Z}^+(s)\cup \mathcal{Z}^+(s')$
	\EndIf
\EndIf
\end{algorithmic}
\end{algorithm}
\vspace{0.5\baselineskip}

\begin{lemma}
Let $p \in \mathcal{C}$. The following equivalence holds:
$p \in \mathcal{P} \text{ if and only if } T = \emptyset \text{ or \texttt{LocationTest}}(p,R'(p,s,T)) \text{ returns \texttt{true} for some }  s\in R(p,T)$
\end{lemma}

\begin{proof}
When $T = \emptyset$ the equivalence is obvious. When $T \neq \emptyset$, we distinguish two cases according to the cardinality of $R'(p,s, T)$: When $|R'(p,s,T)|=1$ it is straightforward that $p\in \mathcal{P} \Leftrightarrow p \in \mathcal{Z}^+(s)$.
When $|R'(p,s,T)|=2$, the configuration of sites in $R'(p,s)$ is like in Figures \ref{fig:sub1} and \ref{fig:sub2}. The test returns \texttt{true} if and only if $p\in \mathcal{P}$ in any case. 
\end{proof}

\begin{figure}
\centering
\begin{subfigure}{.40\textwidth}
  \centering
  \includegraphics{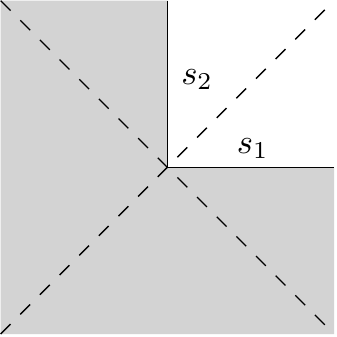}
  \caption{ }
  \label{fig:sub1}
\end{subfigure}%
\begin{subfigure}{.40\textwidth}
  \centering
  \includegraphics{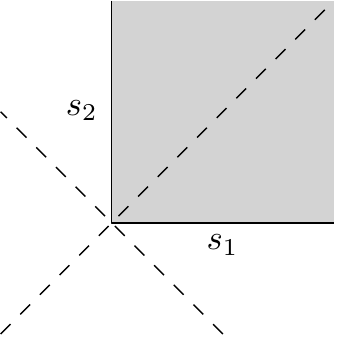}
 \caption{ }
  \label{fig:sub2}
\end{subfigure}

\caption{The two possible configurations of a corner.}
\label{fig:test}
\end{figure}

%\vspace{0.25cm}

%By combining the above, refining a cell and checking termination criteria is done effectively. The construction of a data structure on $\mathcal{P}$ follows.

\subsubsection{Rectangular Decomposition and Bounding Volume Hierarchy.}

We decompose $\mathcal{P}$ into a collection of rectangles such that any two of them have disjoint interior. %We observe that drawing an axis-parallel edge, from every reflex vertex results to a rectangular decomposition.
The resulting decomposition may not be optimal with respect to the number of rectangles but allows for an immediate construction of a data structure (a Bounding Volume Hierarchy) that can answer point and rectangle-intersection queries on the decomposition's rectangles efficiently.
This data structure can accelerate the 2D algorithm for certain inputs and is essential for efficiency of the 3D version of the algorithm.

\paragraph{Rectangular decomposition.}

It is known \cite{Soltan1993} that the minimum number of rectangles in a partition of
a polygon with $n$ vertices and $h$ holes is $n/2+h-g-1$, where $g$ is the maximum number of nonintersecting chords
that can be drawn either horizontally or vertically between reflex vertices. To the direction of minimizing the number of rectangles, the optimal algorithm has $O(n^{3/2}\log n)$ time complexity \cite{Lipski:1984:OMP:1653.1661}. 

However, for our purpose, the decomposition does not need to be optimal. We observe that drawing an axis-parallel edge, from every reflex vertex results to a rectangular decomposition. Thus, we construct a kd-tree on the reflex vertices of the polygon, splitting always at a vertex. Assuming the polygon has $r$ reflex vertices, the kd-tree subdivides the plane into at most $r+1$ regions. Every terminal region contains a disjoint collection of rectangles (nonempty). We denote by $t$ the maximum number of rectangles in a terminal region.

\begin{lemma}
For an orthogonal polygon on $n$ vertices with $h$ holes, let $r$ be the number of its reflex vertices. It holds that: 
$$r = \frac{n}{2}+2(h-1)$$
\end{lemma}

\begin{proof}
A simple orthogonal polygon without holes on $N$ vertices has $N/2 - 2$ reflex vertices. So, if we denote by $n_0$ the number of the outer contour's vertices, and by $n_i$ the respective number for every hole ($i \in [h]$) we have that:
$ r= n_0/2-2 + \sum_{i= 1}^h \big( n_i/2 +2 \big) = \frac{n}{2}+2(h-1)$. 
\end{proof}

\paragraph{Bounding Volume Hierarchy. }

A \textit{Bounding Volume Hierarchy} (BVH) \cite{Haverkort2004} is a tree structure on a set of objects stored at the leaves along with their bounding volume while internal nodes store information of their descendants' bounding volume. Two important properties are \textit{minimal volume} and \textit{small overlap}.

In our setting, the geometric objects are the axis aligned rectangles obtained by the decomposition. Consequently, the most appropriate bounding shape is the Axis Aligned Bounding Box (AABB) offering a good tradeoff between minimal volume and simplicity of representation. As for minimizing the overlap of bounding boxes at the same level in the hierarchy, we group rectangles in a specific way in order to accomplish that:

The BVH is built  in a bottom-up manner, by traversing the kd-tree previously constructed and adding some additional information to its nodes. At every leaf of the kd-tree we compute the AABB of its rectangles (namely a \textit{terminal bounding box}) and for every internal node we find the AABB of its two children. In that way, the bounding volumes of a node's children intersect only at their boundary (see Figure \ref{aabb} for a simple illustration). Space complexity is linear in tree size.

\begin{figure}
\begin{center}
\includegraphics[scale=0.6]{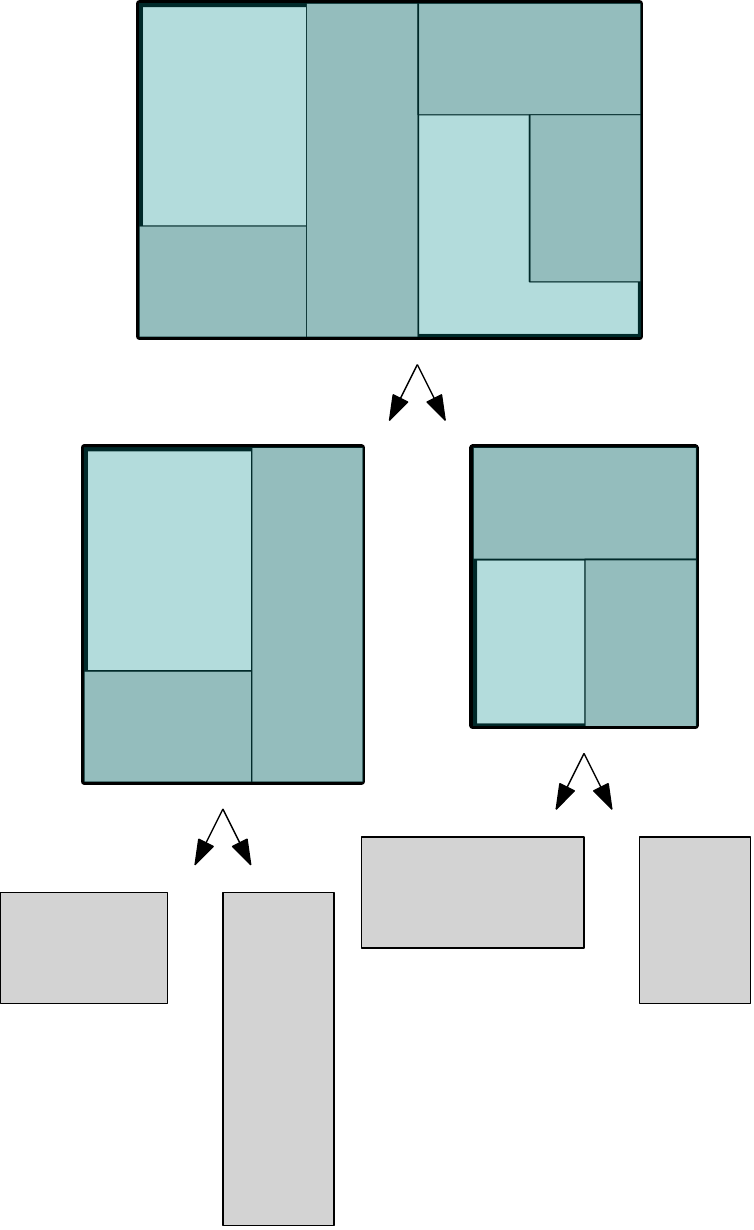}
\end{center}
\caption{A rectangular decomposition for an orthogonal polygon and the corresponding BVH tree \label{aabb}}
\end{figure}

\vspace{0.2cm}
\noindent
\textbf{Rectangle-Intersection queries:}
Given query rectangle $Q$ the data structure reports all rectangles in the decomposition overlapping with $Q$. Starting from the root, for every internal node, we check whether $Q$ intersects its bounding rectangle or not. In the latter case the data structure reports no rectangles. In the former, we check the  position of $Q$ relative to the bounding boxes of the node's children so as to decide for each one if it should be traversed or not. We continue similarly: when we reach a terminal bounding box, we check the position of $Q$ relative to every rectangle in it.
Let $k$ be the number of terminal bounding boxes intersected by $Q$. Following \cite{Agarwal2001BoxtreesAR}, we count the number of internal nodes visited on each level of the tree and show: 

\begin{restatable}{theorem}{bvh}\label{bvh}
Rectangle intersection queries are answered in $O(k\lg r + kt)$.
\end{restatable}

\begin{proof}
Let $Q$ a rectangle-intersection query and $v$ an internal node of the BVH tree visited during the query. We distinguish two cases; first, the subtree rooted at $v$ contains a terminal bounding box that intersects $Q$. There are $O(k)$ such nodes at each level. Otherwise, $Q$ intersects with the bounding rectangle $V$ stored at $v$ but does not intersect any terminal bounding box of the subtree rooted at $v$. There are at least two such terminal bounding boxes, say $b$ and $b'$. Since $Q$ does not intersect $b$ there is a line $\ell$ passing through a facet of $Q$ separating $Q$ from $b$. Similarly, there exists a line $\ell'$ passing through a facet of $Q$ that separates it from $b'$. W.l.o.g. there is a choice of $b, b'$ such that $\ell$ and $\ell'$ are distinct -if all the terminal bounding boxes of the subtree can be separated by the same line, then $V$ cannot intersect $Q$-. If $\ell, \ell'$ are perpendicular, then their intersection also intersects $V$. Since the bounding boxes of each level are strictly non-overlapping, every vertex of $Q$ intersects a constant number of them (up to 4). So, there is a constant number of such nodes at a given level. When $\ell, \ell'$ are parallel and no vertex of $Q$ intersects $V$, then the terminal bounding rectangles of the subtree can be partitioned to those separated by $\ell$ from $Q$ and to those separated by $\ell'$ from $Q$. For these distinct sets of terminal bounding boxes to be formed, there must occur a split of $V$ by a line parallel and in between $\ell, \ell'$. So there is a reflex vertex of the polygon in $V\cap Q$,  causing this split. But $V\cap Q \cap \mathcal{P} = \emptyset$; a contradiction. So there are $O(k)$ internal nodes visited at each level of tree. The visited leaf nodes correspond to the $O(k)$ terminal bounding boxes that intersect $Q$ and since each of them encloses at most $t$ rectangles, the additional amount of performed operations equals $O(kt)$. Summing over all levels of the tree yields a total query complexity of $\sum_{i=0}^{\lceil \lg r \rceil}O(k) + O(kt) = O(k\lg r + kt)$. 
\end{proof}

\noindent
\textbf{Point queries:} Given $p\in \mathbb{R}^2$, we report on the rectangles of the decomposition in which $p$ lies inside (at most~4 rectangles). When zero, the point lies outside the polygon. Since it is a special case of a rectangle-intersection query, the query time complexity is $O(\lg r + t)$.

\medskip\noindent
\textbf{Complexity.}
Analysis requires a bound on the height of the quadtree. The edge length of the initial bounding box is supposed to be~1 under appropriate scaling. %This does not affect arithmetic complexity, but only bit complexity. 
Let \textit{separation bound} $\Delta$ be the maximum value s.t.\ for every cell of edge length $\ge\Delta$ at least one termination criterion holds. Then, the maximum tree height is $L= O(\lg (1/ \Delta))$. Let $\beta$ be the minimum distance of two Voronoi vertices, and $\gamma$ the relative thickness of $\mathcal{P}^c$, i.e.\ the minimum diameter of a maximally inscribed $L_\infty$-ball in $\mathcal{P}^c$, where $\mathcal{P}^c$ is the complement of $\mathcal{P}$.

\begin{lemma}\label{Lsep}
Separation bound $\Delta$ is $ \Omega(\min\{\gamma, \beta\})$, where the asymptotic notation is used to hide some constants. 
\end{lemma}

\begin{proof}
The algorithm mainly subdivides cells that intersect $\mathcal{V}_{\mathcal{D}}(\mathcal{P})$, since a cell inside a Voronoi region or outside $\mathcal{P}$ is not subdivided (Termination criteria (T1), (T2)).
Most subdivisions occur as long as non neighboring Voronoi regions are ``too close". 
Consider $\mathcal{C}$ centered at $p_{_\mathcal{C}} \in V_{\mathcal{D}}(s)$ and $s' \in \widetilde{\phi}(\mathcal{C}) \setminus \phi(\mathcal{C})$, with $V_{\mathcal{D}}(s), V_{\mathcal{D}}(s')$ non neighboring.
For $r_{_\mathcal{C}} < \frac{\mu_\infty(p_{_\mathcal{C}}, s')-\mu_\infty(p_{_\mathcal{C}}, s)}{2} $ site $s'$ is not in $\widetilde{\phi}(\mathcal{C})$.
It holds that $\mu_\infty(p_{_\mathcal{C}}, s')-\mu_\infty(p_{_\mathcal{C}}, s) \ge \zeta(s,s')$, where $\zeta(s,s') = \min\{\mu_\infty(p,q) \mid p \in \overline{V_\mathcal{D}(s)}, q\in \overline{V_\mathcal{D}(s')}\}$, i.e. the minimum distance of the closure of the two Voronoi regions. 
When $\overline{V_\mathcal{D}(s)},\overline{V_\mathcal{D}(s')}$ are connected with a Voronoi edge, $\zeta(s,s') = \Omega(\beta)$. When a minimum cell size of $\Omega(\beta)$ is not sufficient for $s'$ to not belong in $\widetilde{\phi}(\mathcal{C})$, then
there is a hole between $\overline{V_\mathcal{D}(s)},\overline{V_\mathcal{D}(s')}$ and $\Delta$ is $\Omega(\gamma)$ in this case. 
\end{proof}
%(Fig.~\ref{Fex2} for $\Delta \ge \gamma^c$). \qed

This lower bound is tight: in Figure~\ref{Fex1} for $\Delta = 0.8125 \beta$, and in Figure~\ref{Fex2} for $\Delta = \gamma$.
Next we target a realistic complexity analysis rather than worst-case. For this, assume the site distribution in $\mathcal{C}_0$ is ``sufficiently uniform". 
Formally:
\\

\noindent \textbf{Uniform Distribution Hypothesis (UDH):} For $L_\infty$ balls $A_1\subseteq A_0\subset \mathcal{C}_0$, let $N_0$ (resp.\ $N_1$) be the number of sites intersecting $A_0$ (resp.\ $A_1$). We suppose ${N_1}/{N_0} = O({vol(A_1)}/ {vol(A_0)} )$, where $vol(\cdot)$ denotes the volume of a set in $\mathbb{R}^d$, $d$ being the dimension of $\mathcal{C}_0$. 

\begin{theorem}
Under UDH the algorithm complexity is $O( n/ {\Delta} + 1/\Delta^2)$, where $n$ is the total number of boundary edges (including any holes).
\end{theorem}

\begin{proof} At each node, refinement and checking the termination criteria run in time linear in the size of its parent's active set. At the root $|\widetilde{\phi}(\mathcal{C}_0)|=n$. The cardinality of active sets decreases as we move to the lower levels of the quadtree:
Let $A(p,d, R) = \{q\in \mathbb{R}^2 \mid d\le \mu_\infty(p,q) \le 2R+d \}$. For cell $\mathcal{C}$ and $s\in \widetilde{\phi}(\mathcal{C})$, $s \cap A(p_{_\mathcal{C}}, \delta_{_\mathcal{C}}, r_{_\mathcal{C}})\neq \emptyset$. Let $E= vol(A(p_{_\mathcal{C}}, \delta_{_\mathcal{C}}, r_{_\mathcal{C}}))$. For $\mathcal{C}_1$ a child of $\mathcal{C}$ and $s_1\in \widetilde{\phi}(\mathcal{C}_1)$, $s_1 \cap A(p_{_{\mathcal{C}_1}}, \delta_{_{\mathcal{C}_1}}, r_{_{\mathcal{C}_1}})\neq \emptyset$. Since $B_\infty(p_{_\mathcal{C}}, \delta_{_\mathcal{C}})$ is empty of sites and may intersect with $A(p_{_{\mathcal{C}_1}}, \delta_{_{\mathcal{C}_1}}, r_{_{\mathcal{C}_1}})$, we let $E_1=vol(A(p_{_{\mathcal{C}_1}}, \delta_{_{\mathcal{C}_1}}, r_{_{\mathcal{C}_1}})\setminus (A(p_{_{\mathcal{C}_1}}, \delta_{_{\mathcal{C}_1}}, r_{_{\mathcal{C}_1}}) \cap B_\infty(p_{_\mathcal{C}}, \delta_{_\mathcal{C}})))$. We prove that in any combination of $\delta_{_\mathcal{C}}, \delta_{_{\mathcal{C}_1}}, r_{_\mathcal{C}}$ it is $ E_1 \le E/2$. 
Under Hypothesis~1, a cell at tree level $i$  has $|\widetilde{\phi}(\mathcal{C}_i)|= O(n/2^i)$. %it follows that at tree level $i$ computation time is $\Theta(2^in)$;  %$|\widetilde{\phi}(\mathcal{C}_1)| \le |\widetilde{\phi}(\mathcal{C})|/2$. Therefore, since
Computation per tree level, is linear in sum of active sets' cardinality, therefore
%it is linear in $\sum_{\mathcal{C}}|\widetilde{\phi}(\mathcal{C})|$ for $\mathcal{C}$ at level $i$. S
summing over all levels of the tree, we find the that complexity of the subdivision phase is $O(n/\Delta)$. The complexity of the reconstruction phase is $O(\tilde{n})$, where $\tilde{n}$ is the number of leaf nodes in the quadtree, which is in turn $O(1/\Delta^2)$. This allows to conclude. \end{proof}
%\vspace{3mm}

Queries in the BVH can be used to compute label sets and the active set of a cell. Assume the number of segments touching a rectangle's boundary is $O(1)$, which is the typical case. Then, we prove the following.

\begin{figure}[h]
\begin{subfigure}{0.45\textwidth}
  \centering
  \includegraphics[scale=0.3]{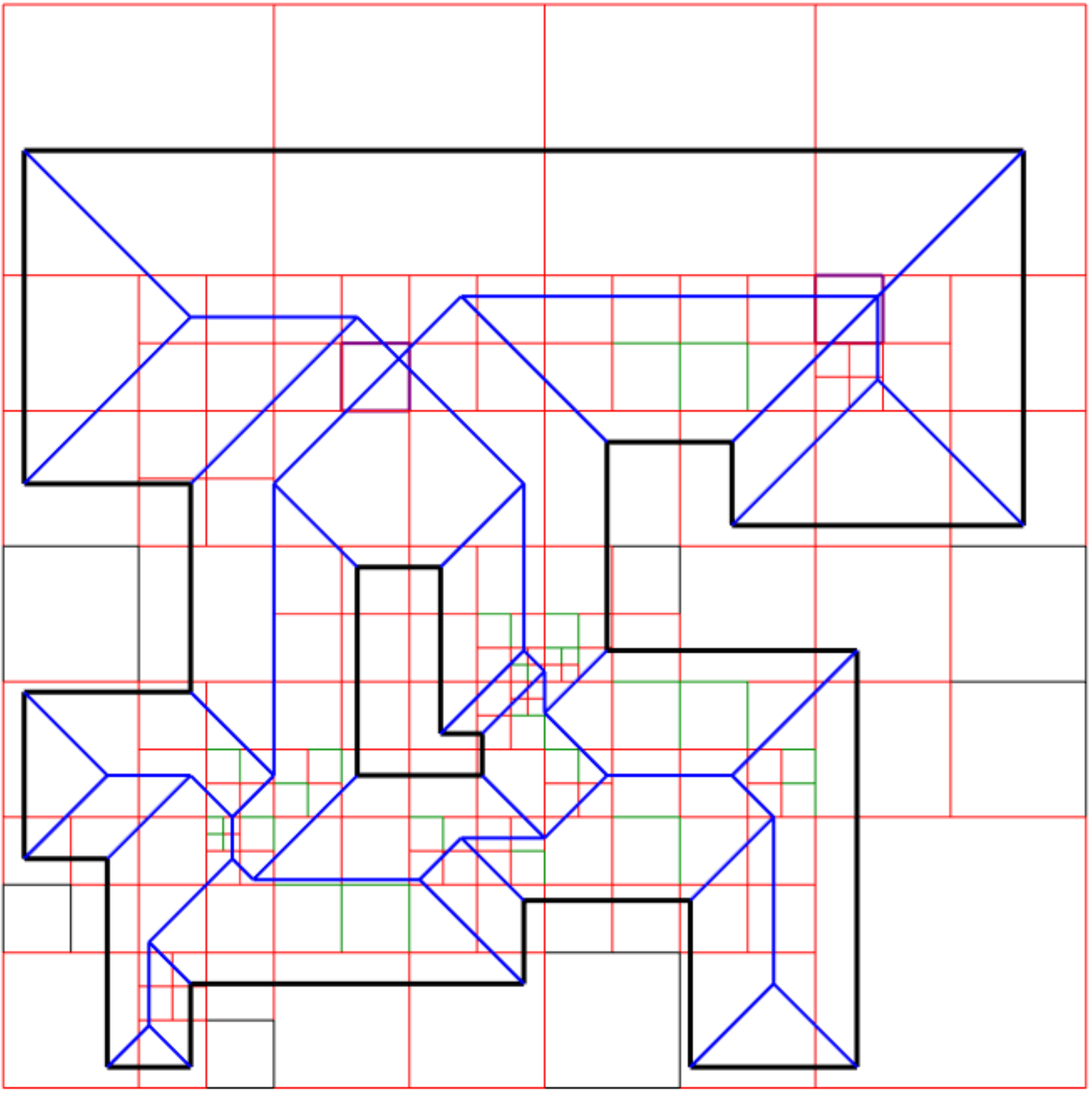}
  \caption{Input consists of 28 sites and 164 cells are generated. Total time is 12.0 ms. Minimum cell size is $0.8125 \cdot \beta$. \label{Fex1}}
\end{subfigure}%
\hspace{0.2cm}
\begin{subfigure}{0.45\textwidth}
  \centering
  %\vspace{-0.45cm}
  \includegraphics[scale=0.245]{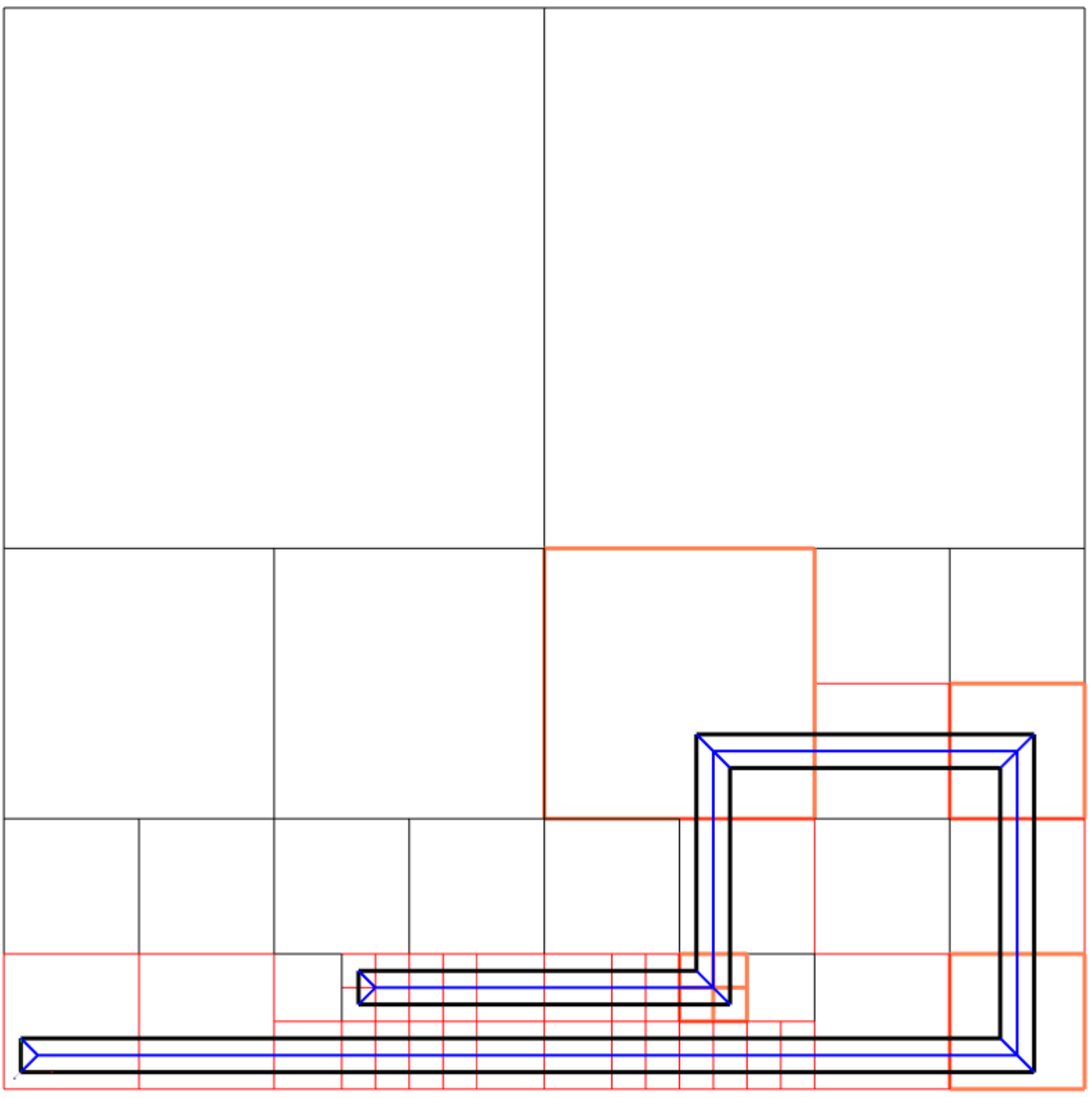}
 \caption{Input consists of 12 sites and 64 cells are generated. Total time is 5.6 ms. Minimum cell size is $\gamma$.\label{Fex2} }
 \end{subfigure}
 \caption{The 1-skeleton of the Voronoi diagram is shown in blue. \label{Fex}}
\end{figure}

\begin{lemma}\label{Lbvh}
We denote by $\mathcal{C}'$ the parent of $\mathcal{C}$ in the subdivision.
Using BVH accelerates the refinement of $\mathcal{C}$
if ${|\widetilde{\phi}(\mathcal{C}')|} / {|\widetilde{\phi}(\mathcal{C})|} = \Omega( \lg n+t)$.
\end{lemma}

%\begin{corollary}
%The time complexity of the 2D subdivision algorithm using the BVH data structure is $O(\frac{3\widetilde{L}}{\sqrt{2}}(\frac{1}{\Delta} - \frac{1}{L_0})k\log n)$.

%\noindent 
%Assume the number of segments touching a rectangle's boundary is $O(1)$, which is the typical case.
%\vspace{3mm}

\begin{proof}
 A label set $\lambda(p)$ is determined by performing a point and a rectagle-intersection query; once the point is detected to lie inside a rectangle $R_0$ of a leaf $T_0$ we find an initial estimation $d_0$ of $\mu_\infty(p,\mathcal{P})$. Since the closest site to $p$ may be on another leaf, we do a rectangle-intersection query centered at $p$ with radius $d_0$. The closest site(s) to $p$ are in the intersected leaves. Thus, finding $\lambda(p)$ takes $O(k(p,d_0)\lg r + k(p,d_0)t)$ (Thm.~\ref{bvh}), where $k(p,d_0)= O(1)$ is the number of BVH leaves intersected by $\overline{B_\infty(p,d_0)}$. 
%A label set $\lambda(p)$ is determined by performing a point and a rectagle-intersection query; once the point is detected to lie inside
%a rectangle $R_0$ of a leaf $T_0$ we compute its closest sites(s) and the respective distance as follows:
%first, $\mu_\infty(p,\mathcal{P})$ is estimated by calculating the distance of $p$ to the sites
%on the boundary of $R_0$. The number of sites that touch the boundary of each rectangle is $O(1)$. Let $d_0$ be the initial estimation. Since the closest site to $p$ may be on the boundary of another
%rectangle of the decomposition, we must search in leaves of the BVH tree at distance $d_0$ from $p$. These leaves are determined by a rectangle-intersection query. Then we visit the rectangles stored at the intersected leaves and update our estimation on $\mu_\infty(p, \mathcal{P})$ and its closest sites accordingly.
%Thus, the complexity of finding $\lambda(p)$ is $O(k(p,d_0)\log r + k(p,d_0)t)$ (Thm.~\ref{bvh}), where $k(p,d_0)$ is the number of BVH leaves intersected by $\overline{B_\infty(p,d_0)}$. However $k(p,d_0)=O(1)$.
Computing the sites in $\widetilde{\phi}(\mathcal{C})$ is accelerated if combined with a rectangle intersection query to find segments at $L_\infty$-distance $\le 2r_{_\mathcal{C}}+\delta_{_\mathcal{C}}$ from $p_{_\mathcal{C}}$. 
Let $k_{_\mathcal{C}}$ be the maximum number of BVH leaves intersected by this rectangle-intersection query.
We obtain a total refinement time for the cell equal to $O(k_{_\mathcal{C}}\lg r + k_{_\mathcal{C}}t)$. %, where $k_{_\mathcal{C}}$ is the maximum number of BVH leaves intersected by a rectangle-intersection query during the refinement of $\mathcal{C}$.
Since $k_{_\mathcal{C}}= O(|\widetilde{\phi}(\mathcal{C})|)$ and $r = O(n)$ the lemma follows.
\end{proof}

%------------------------------------------

\section{Subdivision algorithm in three dimensions}\label{S3d}

Let $\mathcal{P}$ be a manifold orthogonal polyhedron: every edge of $\mathcal{P}$ is shared by exactly two and every vertex by exactly 3 facets.  %\cite{Jare-1991}. 
For non-manifold input we first employ trihedralization of vertices, discussed in \cite{DBLP:journals/cvgip/MartinezGA13,straight3d}. Input consists of $\mathcal{S}$ and bounding box $\mathcal{C}_0$ of $\mathcal{P}$. An octree is used to represent the subdivision. 
%The sites $\mathcal{S}$ consist of the facets of the polyhedron, considered as closed sets. 
%We assume that no two facets with the same affine hull induce a 3D bisector in the diagram (otherwise we infinitesimally perturb them).

The main difference with the 2D case is that Voronoi sites can be nonconvex. As a consequence, for site $s$, $\mathcal{Z}^+(s)$ is not necessarily convex and therefore the distance function $D_s(\cdot)$ cannot be computed in $O(1)$ time: it is not trivial to check membership in $\mathcal{Z}^+(s)$. It is direct to extend the 2D algorithm in three dimensions. However, we examine efficiency issues.

For an efficient computation of the basic predicates (of Chapter~\ref{SSpredicates}), we \textbf{preprocess} every facet of the polyhedron and decompose it to a collection of rectangles. Then a BVH on the rectangles is constructed. The basic operation of all these predicates in 2D is an overlap test between an interval and a segment in 1D. In 3D, the analog is an overlap test between a 2D rectangle and a site (rectilinear polygon). Once the BVH is constructed for each facet, the rectangle-intersection query takes time logarithmic in the number of facet vertices (Theorem~\ref{bvh}). 

\subsubsection{Subdivision.} 
The active set $\widetilde{\phi}$, $\phi$ and the label set of a point are defined as in to 2D. Most importantly, Lemma~\ref{lem:yap} is valid in 3D as well. % (proof in Appendix). 
The algorithm proceeds as follows: We recursively subdivide $\mathcal{C}_0$ into 8 identical cells. The subdivision of a cell stops whenever at least one of the termination criteria below holds. For each cell of the subdivision we maintain the label set of its central point and $\widetilde{\phi}$. Upon subdivision, we propagate $\widetilde{\phi}$ from a parent cell to its children for further refinement. We denote by $M$ the maximum degree of a Voronoi vertex $(M\le 24)$.

\textbf{3D Termination criteria:} (T1') $int(\mathcal{C} ) \cap \mathcal{P} = \emptyset$, 
(T2') $|\widetilde{\phi}(\mathcal{C} ) | \le 4$, 
(T3') $|\widetilde{\phi}(\mathcal{C} ) | \le M$ and the sites in $\widetilde{\phi}(\mathcal{C} ) $ define a unique Voronoi vertex $v \in \mathcal{C}$.

Subdivision is summarized in Alg.~\ref{algo2}.
(T1') is valid for $\mathcal{C}$ iff $\lambda(p_{_\mathcal{C}}) = \emptyset$ and $\forall s\in \widetilde{\phi}(C)$ it holds that $s\cap int(\mathcal{C}) =\emptyset$. Detecting a Voronoi vertex in $\mathcal{C}$ proceeds like in 2D. A Voronoi vertex is equidistant to at least 4 sites and there is a site parallel to each coordinate hyperplane among them.

(T1) used in 2D is omitted, for it is not efficiently decided: labels of the cell vertices cannot guarantee that $\mathcal{C}\subseteq V_\mathcal{D}(s)$. However, as the following lemma indicates, termination of the subdivision is not affected: cells contained in a Voronoi region whose radius is $<r^*$, where $r^*$ is a positive constant, are not subdivided. 

\begin{restatable}{lemma}{Lrad} \label{Lrad}
Let $\mathcal{C}\subseteq V_{\mathcal{D}}(s)$. There exists $r^*>0$ s.t. if  $r_{_{\mathcal{C}}}<r^*$ it holds that $\widetilde{\phi}(\mathcal{C}) = \{s\}.$
\end{restatable}

\begin{proof}
Let $s' \in \mathcal{S}\setminus s$. We will prove that if $\mathcal{Z}^+(s')\cap\mathcal{C} \neq \emptyset$, it holds that $\delta_{_\mathcal{C}} < \mu_\infty(p_{_\mathcal{C}}, s')$. Therefore there is $r(s')>0$ such that $2r(s') +\delta_{_\mathcal{C}} < \mu_\infty(p_{_\mathcal{C}}, s')$. Let $r^*$ be the minimum of these radii for every site different than $s$. When $r_{_{\mathcal{C}}}<r^*$,  it holds that $\widetilde{\phi}(\mathcal{C}) = \{s\}.$

Suppose that $\delta_{_\mathcal{C}} = \mu_\infty(p_{_\mathcal{C}}, s') = \mu_\infty(p_{_\mathcal{C}}, q)$ for $q\in s'$. Then $q\in  \text{aff}(s)$ and $s'$ cannot be a subset of $\text{aff}(s)$. So, $s'$ is adjacent to $s$. If $\mathcal{Z}^+(s')\cap \mathcal{C} = \emptyset$ then $s'\not \in \widetilde{\phi}(\mathcal{C})$. If $\mathcal{Z}^+(s')\cap \mathcal{C} \neq \emptyset$, since for adjacent sites it holds that $\mathcal{Z}^+(s)\cap \mathcal{Z}^+(s') = \text{bis}_\mathcal{D}(s,s')$, bisector intersects $\mathcal{C}$ which is a contradiction. Thus, there is no site $s'$ with $\mathcal{Z}^+(s')\cap\mathcal{C} \neq \emptyset$ and $\mu_\infty(p_{_\mathcal{C}}, s')=\delta_{_\mathcal{C}}$. 
\end{proof}

\renewcommand{\thealgorithm}{2}
\begin{algorithm}
	\caption{Subdivision3D($\mathcal{P}$)}
	\label{algo2}
	\begin{algorithmic}[1]
		%\Require  {orthogonal polygon $\mathcal{P}$}
                    %  \Ensure {}
\State root $\gets$ bounding box of $\mathcal{P}$
\State $Q \gets$ root
\While {$Q \neq \emptyset$}
	\State $\mathcal{C} \gets$ pop$(Q)$
        \State Compute the label set of central point and $\widetilde{\phi}(\mathcal{C})$.
	\If {(T1') $\lor$ (T2') $\lor$ (T3')}
		\State \Return
	\Else 
\State Subdivide $\mathcal{C}$ into $\mathcal{C}_1, \dots, \mathcal{C}_8$
\State $Q \gets$ $Q \cup \{ \mathcal{C}_1,\dots, \mathcal{C}_8\}$
	\EndIf
\EndWhile
\end{algorithmic}
\end{algorithm}

\begin{restatable}{theorem}{threedhalts}\label{Tterm}
Algorithm~\ref{algo2} halts.
\end{restatable}

\subsubsection{Reconstruction.} 
We construct a graph $G=(V,E)$, representing the 1-skeleton of the Voronoi diagram. The nodes of $G$ are of two types, \textit{skeleton nodes} and \textit{Voronoi vertex nodes}, and are labeled by their closest sites. Skeleton nodes span Voronoi edges and are labeled by 3 or 4 sites. We visit the leaves of the octree and process cells with $|\widetilde{\phi}(\mathcal{C})| \ge 3$ and that do not satisfy (T1'). We %proceed as in 2D Reconstruction phase, 
introduce the nodes to the graph as in 2D. Graph edges are added between corners and Voronoi vertex nodes inside a cell. %The remaining graph edges must cross two subdivision cells. We use the dual marching cubes method \cite{Scot-2005} to efficiently enumerate all tuples of neighboring subdivision cells in time linear in the size of the octree (the \texttt{cellProc} and \texttt{faceProc} procedures of the method are used).
We run dual marching cubes (linear in the octree size) and connect graph nodes $v_1,v_2$ located in neighboring cells, iff:
%
% \begin{itemize} \item
\\\hspace*{2mm} $-$ $v_1, v_2$ are skeleton nodes and $\lambda(v_1)=\lambda(v_2)$, or
    \\\hspace*{2mm} $-$ $v_1$ is a skeleton  node, $v_2$ is a Voronoi vertex node and $\lambda(v_1)\subset \lambda(v_2)$, or 
    \\\hspace*{2mm} $-$ $v_1, v_2$ are Voronoi vertex nodes, $\lambda(v_1)\cap \lambda(v_2)= \{s, s', s''\}$ and $v_1v_2\subset \mathcal{P}$.
% \end{itemize}

\begin{restatable}[Correctness]{theorem}{reconstruction}
The output graph is isomorphic to the 1-skeleton of the Voronoi diagram.
\end{restatable} 

\noindent
{\bf Primitives.}
Deciding membership in $\mathcal{H}(\cdot)$ is trivial.
 The predicates of Sect.~\ref{SSpredicates} extend to 3D and the runtime of each is that of a rectangle-intersection query on the BVH constructed for the corresponding site at preprocessing: 
 Let $pr_{\text{aff}(s)}(p)$ be the projection of $p$ to aff$(s)$ and $B_{p,s}$ the $2d-$box on $\text{aff}(s)$ centered at $pr_{\text{aff}(s)}(p)$ with radius $\mu_\infty(p,\text{aff}(s))$. 
Then, $p\in \mathcal{Z}(s)$ iff $B_{p,s} \cap s \neq \emptyset$. A query with $B_{p,s}$ is done by \textbf{\texttt{inZone($p,s$)}}. % and \textbf{\texttt{inOrientedZone($p,s$)}}. 
For \textbf{\texttt{ZoneInCell($p,\mathcal{C}$)}} we do a %rectangle-intersection 
query with $\overline{B_\infty(p_{_\mathcal{C}}', r_{_\mathcal{C}}+\rho_i)}$ where $p_{_\mathcal{C}}'=pr_{\text{aff}(s)}(p_{_\mathcal{C}})$, $\rho_i = \mu_\infty(f_i, \text{aff}(s))$ and $f_1,f_2$ the two facets of $\mathcal{C}$ parallel to aff$(s)$. Queries with $ \overline{B_\infty(p_{_\mathcal{C}}', r_{_\mathcal{C}})}$ are also performed by \textbf{\texttt{isIntersecting($s,\mathcal{C}$)}} and \textbf{\texttt{isSt\-rictlyIntersecting($s,\mathcal{C}$)}}. When computing \textit{label sets}, \texttt{LocationTest} is slightly modified, since the corners used to identify the position of a point can also be formed by 3 sites.
 
\medskip\noindent
{\bf Complexity.} Under appropriate scaling so that the edge length of $\mathcal{C}_0$ be 1, if $\Delta$ is the separation bound, then the maximum height of the octree is $L= O(\lg (1/ \Delta))$. The algorithm mainly subdivides cells intersecting $\mathcal{P}$, unlike the 2D algorithm that mainly subdivides cells intersecting $\mathcal{V}_\mathcal{D}(\mathcal{P})$, because a criterion like (T1) is missing. This absence does not affect tree height, since by Lemma~\ref{Lrad} the minimum cell size is same as when we separate sites whose regions are non-neighboring (handled by Lemma~\ref{Lsep}). If $\beta$ is the minimum distance of two Voronoi vertices and $\gamma$ the relative thickness of $\mathcal{P}^c$, taking $\Delta = \Omega(\min\{\beta, \gamma\})$ suffices, as in 2D (Lemma~\ref{Lsep}).
%the absence of (T1) does not further affect tree height, since from Lem.~\ref{Lrad} this falls into the case of separating sites with non-neighboring Voronoi regions. 

%we let $r^* = \min_{s\in \mathcal{S}}(r^*(s))$. Then the separation bound is $\Delta = O(\min\{ r^*, \gamma, \beta\})$.

\begin{theorem}\label{Tcompl}
Under UDH and if $n$ is the number of polyhedral facets, $\alpha$ the maximum number of vertices per facet and $t_\alpha$ the maximum number of rectangles in a BVH leaf, the algorithm's complexity is 
$O( n\, \alpha (\lg\alpha +t_\alpha)/\Delta^{2} + 1/\Delta^3)$.
\end{theorem}

\begin{proof} We sum active sets' cardinalities of octree nodes, since refining a cell requires a number of rectangle-intersection queries linear in the size of its parent's active set. Let $\mathcal{C}$ and its child $\mathcal{C}_1$. Any $s\in \widetilde{\phi}(\mathcal{C})$ satisfies $ \delta_{_\mathcal{C}}\le \mu_\infty(p_{_\mathcal{C}},s) \le 2r_{_\mathcal{C}}+\delta_{_\mathcal{C}}$. We denote by $E$ the volume of the annulus $\{q\in \mathbb{R}^2 \mid \delta_{_\mathcal{C}}\le \mu_\infty(p_{_\mathcal{C}},q) \le 2r_{_\mathcal{C}}+\delta_{_\mathcal{C}} \}$ and by $E_1$ the volume of the respective annulus for $\mathcal{C}_1$, minus the volume of the annulus' intersection with $B_\infty(p_{_\mathcal{C}}, \delta_{_\mathcal{C}}) $. It is easy to show $E_1\le E/2 $. Under Hypothesis~1, we sum all levels and bound by %$O(\sum_{i=1}^L (4^in))=
$O(4^Ln)$ the number of rectangle-intesection queries. Using Thm.~\ref{bvh}, we find that the complexity of the subdivision phase is $O( n\, \alpha (\lg\alpha +t_\alpha)/\Delta^{2})$. The complexity of the reconstruction phase is $O(\tilde{n})$, where $\tilde{n}$ is the number of leaf nodes in the octree, which is in turn $O(1/\Delta^3)$. This allows to conclude. 
\end{proof}

\begin{remark}\label{Rmk3d}
$t_\alpha=O(\alpha)$ so the bound of Thm.~\ref{Tcompl} is $O(n\alpha^2/\Delta^{2}+1/\Delta^3)$.
Let $V$ be the number of input vertices.
It is expected that $n\alpha = O(V)$; 
also $\alpha$ % $ = O(V)$ 
is usually constant. 
In this case, the complexity simplifies to $O(V/\Delta^{2} + 1/\Delta^3)$.
\end{remark}

\vspace{-0.6cm}

%\caption{The 1-skeleton of the Voronoi diagram is shown in blue. \label{Fexamples}}
%\end{figure}

%\begin{remark} 
%Let $E_s$ be the volume of $V_\mathcal{D}(s)$ and $E_{min}= \min_{s\in \mathcal{S}}E_s$. Since at level $i$ cells have size $2^{-i}$, there are $\Omega(2^{3i}E_s)$ cells s.t.\ $s\in \widetilde{\phi}(\cdot)$, and $\Omega(2^{3i}E_s)$ queries at BVH corresponding to $s$ are performed (in time $q_s=\Omega(\log\alpha_s)$, where $\alpha_s$ the number of vertices of $s$).
% (each in $q_j= \Omega(\log_2\alpha_j)$, where $\alpha_j$ the number of facet's vertices). 
%Summing for all levels and sites, the complexity of the algorithm is
%$ \sum_{i=0}^{log_2(1/\Delta)} \sum_{j=1}^n \Omega(2^{3i}E_jq_j) =
%$\Omega (E_{min}\lg V/\Delta^3).$

%\end{remark}
\begin{figure}
  \centering
  \includegraphics[scale=0.15]{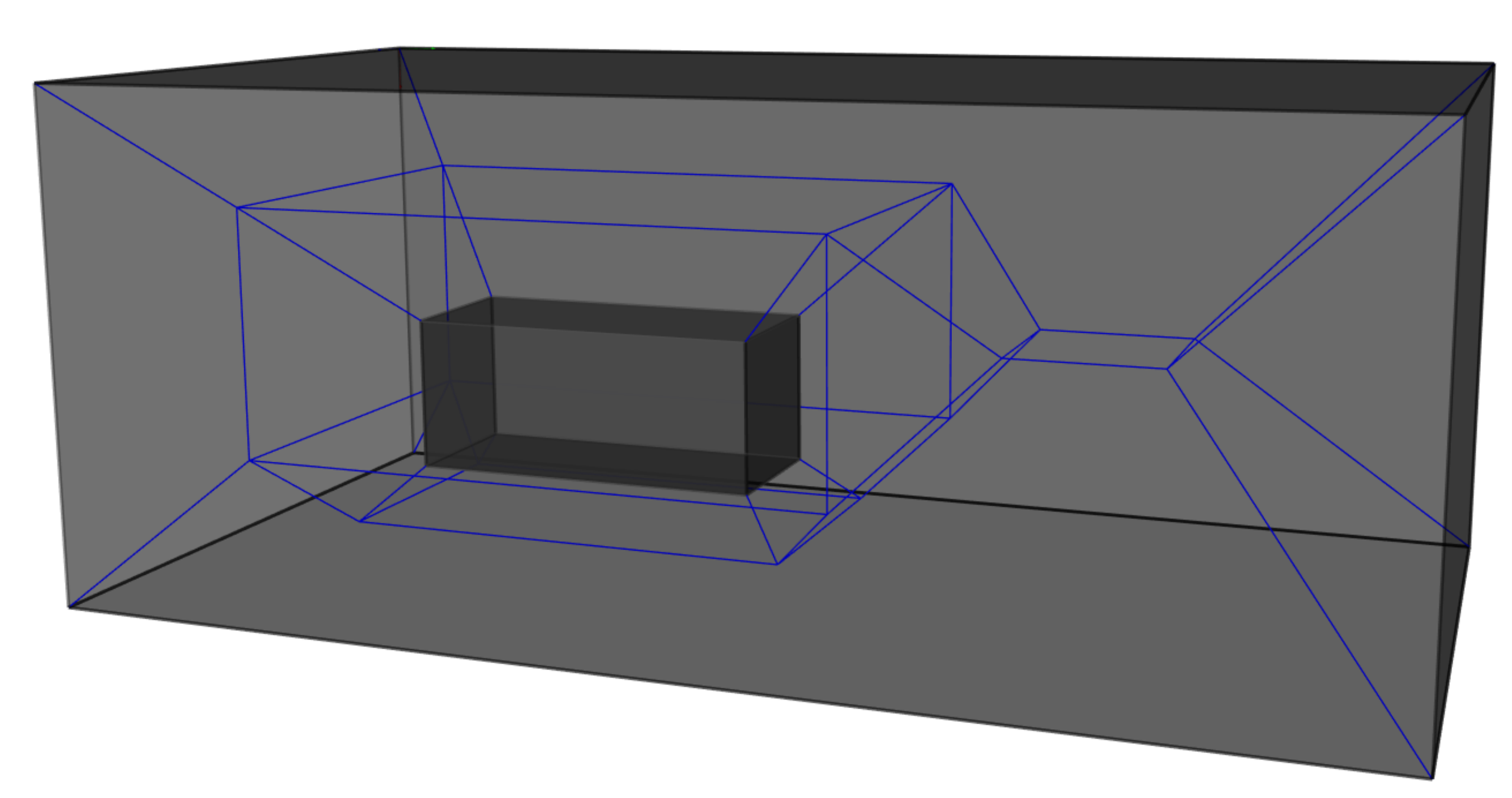}
 \caption{The 1-skeleton of the Voronoi diagram is shown in blue. Input consists of 12 sites and 386 cells are generated (not shown). Total time is 94.8 ms.\label{Fex3}}
\end{figure}

\section{Implementation and concluding remarks}\label{SImpl}

Our algorithms were implemented in Julia %(\url{https://julialang.org/}) 
and use the algebraic geometric modeler \texttt{Axl} %(\url{http://axel.inria.fr/}) 
for visualization. 
 They are available in \url{https://gitlab.inria.fr/ckatsama/L\_infinity\_Voronoi}. %They consist of approximately 5000 lines of code.
 They are efficient in practice and their performance scales well in thousands of input sites.
Some examples with runtimes are given in Figure~\ref{Fex} and~\ref{Fex3}. All experiments were run on a 64-bit machine with an Intel(R) Core(TM) i7-8550U CPU @1.80GHz and 8.00~GB of RAM.

Our complexity bounds rely on a hypothesis of "sufficiently uniform" input. 
Even though Uniform Distribution Hypothesis might seem a strong assumption, note that our analysis does not encounter the width of the subdivision tree and considers it as a complete tree. 
We anticipate that in instances where our hypothesis does not hold, the ‘‘non-uniform" distribution of sites is captured by sudden increases or decreases in the width of the tree passing from a higher level to a lower one. 
We expect that there is a trade-off between the number of cells at level $i$ of the subdivision tree variating from $4^i$ and the cardinality of the active set of a cell being independent of the current level, so that the total complexity remains linear in the number of sites (facets).
We plan to extend our research towards proving non trivial complexity bounds that exploit to the maximum degree the adaptivity of subdivision algorithms and do not rely on assumptions.

\vspace{0.2cm}

\noindent
{\small{\bf Acknowledgements.}
We thank Evanthia Papadopoulou for commenting on a preliminary version of the paper and Bernard Mourrain for collaborating on software. 
Both authors are members of AROMATH, a joint team between INRIA Sophia-Antipolis (France) and NKUA.}
\vfill
%\pagebreak
\bibliographystyle{splncs04}
\bibliography{mybibliography}

\newpage

\end{document}